\documentclass[12pt]{article}
\usepackage[english]{babel} % idioma del documento
\usepackage[T1]{fontenc} % encoding del .pdf
\usepackage{lmodern} % para T1, la fuente default (Computer Modern) no funciona, hay que usar lmodern (por ejemplo)
\usepackage{amsmath}
\usepackage{amssymb}
\usepackage{amsthm}
\usepackage{mathrsfs} % to use \mathscr, a "very cursive" math font
\allowdisplaybreaks[1] % (requiere amsmath) permitir saltos de pagina en ecuaciones (1: evitarlos si se puede, 4: se permiten siempre)
\usepackage[a4paper, margin=1in]{geometry}
\usepackage{hyperref}
\usepackage{xcolor}

\def\RR{\mathbb{R}} % reals
\def\NN{\mathbb{N}} % naturals
\def\PP{\mathbb{P}} % probability
\def\EE{\mathbb{E}} % expectation
\def\ind{\mathbf{1}} % indicator
\def\law{\mathcal{L}} % law
\def\P{\mathcal{P}} % probability measures
\def\W{\mathcal{W}} % Wasserstein distance
\def\N{\mathcal{N}} % Poisson point measure
\def\M{\mathcal{M}} % another Poisson point measure
\def\F{\mathcal{F}} % Filtration
\def\C{\mathcal{C}} % Set G
\def\G{\mathcal{G}} % Set G
\def\i{\mathbf{i}} % function i
\def\ii{\hat{\imath}} %
\def\jj{\hat{\jmath}} %
\def\SS{\mathbb{S}} % unit sphere
\newcommand{\st}[1]{[ #1 ]}

\newtheorem{thm}{Theorem}
\newtheorem{lem}[thm]{Lemma}
\newtheorem{cor}[thm]{Corollary}
\newtheorem{defi}[thm]{Definition}
\theoremstyle{definition}
\newtheorem{rmk}[thm]{Remark}

\title{Quantitative uniform propagation of chaos for Maxwell molecules}

\author{Roberto Cortez\footnote{CIMFAV, Facultad de Ingenieria, Universidad de Valpara\'iso, General Cruz 222, Valpara\'iso, Chile.  \ E-mail:   \texttt{rcortez@dim.uchile.cl}} \,  and  \, Joaquin Fontbona\footnote{Department of Mathematical Engineering and Center for Mathematical Modeling,  UMI(2807) UCHILE-CNRS, University of  Chile,  { Casilla 170-3, Correo 3, Santiago-Chile.}  E-mail:  \texttt{fontbona@dim.uchile.cl}.  } }

\begin{document}

\maketitle

%%%%%%%%%%%%%%%%%%%%%%%%%%%%%%%%%%%%%%%%%%%%%%%%%%%%%%%%
%%%%%%%%%%%%%%%%%%%%%%%%%%%%%%%%%%%%%%%%%%%%%%%%%%%%%%%%
\begin{abstract}
%%%%%%%%%%%%%%%%%%%%%%%%%%%%%%%%%%%%%%%%%%%%%%%%%%%%%%%%
%%%%%%%%%%%%%%%%%%%%%%%%%%%%%%%%%%%%%%%%%%%%%%%%%%%%%%%%

%\noindent 
We prove propagation of chaos at explicit polynomial rates in Wasserstein distance $\W_2$  for  Kac's $N$-particle system   associated with the spatially homogeneous Boltzmann equation for Maxwell molecules. Our approach is mainly  based on novel probabilistic coupling techniques. Combining them with recent stabilization results for the particle system we obtain, under suitable moments assumptions on the initial distribution, a  uniform-in-time estimate  of order almost $N^{-1/3}$  for $\W^2_2$.
\end{abstract}

\medskip


%%%%%%%%%%%%%%%%%%%%%%%%%%%%%%%%%%%%%%%%%%%%%%%%%%%%%%%%
%%%%%%%%%%%%%%%%%%%%%%%%%%%%%%%%%%%%%%%%%%%%%%%%%%%%%%%%
\section{Introduction and main results}
%%%%%%%%%%%%%%%%%%%%%%%%%%%%%%%%%%%%%%%%%%%%%%%%%%%%%%%%
%%%%%%%%%%%%%%%%%%%%%%%%%%%%%%%%%%%%%%%%%%%%%%%%%%%%%%%%

\subsection{The Boltzmann equation}
The \emph{spatially homogeneous Boltzmann equation} predicts that the density $f_t(v)$ of particles with velocity $v\in\RR^3$ at time $t\geq 0$ in a spatially homogeneous dilute gas subjected to binary collisions, satisfies
\begin{equation}
\label{eq:B}
\partial_t f_t(v) = \frac{1}{2} \int_{\RR^3} dv_*
\int_{\mathbb{S}^2} d\sigma [f_t(v')f_t(v_*') - f_t(v)f_t(v_*)] B(|v-v_*|,\theta),
\end{equation}
where $v'$ and $v_*'$ are the \emph{pre-collisional velocities}, given by
\begin{equation}
\label{eq:pre_col}
v' = \frac{v+v_*}{2} + \frac{|v-v_*|}{2}\sigma,
\quad
v_*' = \frac{v+v_*}{2} - \frac{|v-v_*|}{2}\sigma,
\end{equation}
and $\theta$ is the \emph{deviation angle}, defined by $\cos\theta = \sigma \cdot (v-v_*)/|v-v_*|$. The \emph{collision kernel} $B(|v-v_*|,\theta)\geq 0$ describes the rate at which collisions between pairs of particles occur, and depends on the type of physical interactions among them. Solutions of \eqref{eq:B} preserve mass $\int_{\RR^3} f_t(v)dv$, momentum $\int_{\RR^3} v f_t(v)dv$ and kinetic energy $\int_{\RR^3} |v|^2 f_t(v)dv$, so we may and shall assume that $\int_{\RR^3} f_t(v)dv=1$ for all $t\geq 0$.

Equation \eqref{eq:B} has been extensively studied for several decades. We refer the reader to Cercignani \cite{cercignani1988} for physical  background on the  Boltzmann equation and to Villani \cite{villani2002}, Alexandre \cite{alexandre2009} and Mischler and Mouhot \cite{mischler-mouhot2013} for historical accounts on aspects of its  mathematical theory.

Typically, one assumes that $B:\RR_+\times (0,\pi]\to \RR_+$ has the form
\[
B(z,\theta) \sin\theta = z^\gamma \beta(\theta),
\]
for some $\gamma \in(-3,1]$, and some  function $\beta:(0,\pi]\to \RR_+$ which, for symmetry reasons, can be taken to be equal to $0$ on $(\pi/2,\pi]$. In this paper we assume that $\gamma=0$ and $\beta(\theta)\sim \theta^{-1-\nu}$ near 0 for $\nu=1/2$, a setting referred to as \emph{Maxwell molecules} case.

\subsection{Particle system and propagation of chaos}

As a step to rigorously justify the interpretation of the Boltzmann equation as  a representation of  the evolution of a very large number of interacting particles, Kac  \cite{kac1956} suggested to study  the limit, as $N$ goes to $\infty$, of some exchangeable stochastic   system of $N$ of such particles, defined as a continuous-time  pure-jump Markov process on $(\RR^3)^N$. For a simplified one dimensional version of the nonlinear equation \eqref{eq:B}, he in fact proved that if the  joint law of $k$ particles at time $0$  weakly converges as $N\to\infty$ to the $k$-fold product of an initial density $f_0$ in $\RR$, then the same holds true at times $t> 0$, with $f_t$ as the limit density.

This property, termed \emph{propagation of chaos}, is equivalent  to the \emph{convergence of the  empirical measure of the system at each time $t$ to the solution of the nonlinear equation}, and has been established,  under different convergence criteria, for a wide class of models including the true  Boltzmann equation  \eqref{eq:B}. For general background on  propagation of chaos we refer the reader to Sznitman \cite{sznitman1989}, M\'{e}l\'{e}ard  \cite{meleard1996} and Mischler and Mouhot \cite{mischler-mouhot2013} (see also Section \ref{known results}  below for  historical and recent references).

More specifically, we consider the particle system given by the $(\RR^3)^N$-valued Markov process  with infinitesimal generator $\mathcal{A}^N$ defined as follows: for all Lipschitz bounded  function $\Phi:(\RR^3)^N \to \RR$ and  $\mathbf{v}=(v^1,\ldots,v^N) \in (\RR^3)^N$, 
\begin{equation}
\label{eq:gen}
\mathcal{A}^N \Phi(\mathbf{v})
= \frac{1}{2(N-1)} \sum_{i\neq j} \int_{\SS^2} d\sigma
[\Phi(\mathbf{a}_{ij}(\mathbf{v},\sigma))
- \Phi(\mathbf{v})] B(\theta) ,
\end{equation}
where $B(\theta) \sin \theta := \beta(\theta)$, and $\mathbf{a}_{ij}(\mathbf{v}, \sigma) \in (\RR^3)^N$  is the vector $\mathbf{v}$ with its $i$-th and $j$-th $\RR^3$-valued components respectively replaced by $v'$ and $v_*'$ given by \eqref{eq:pre_col} with $v=v^i$ and $v_*=v^j$. Particles start with a  symmetric law $G_0^N$ on $(\RR^3)^N$. We denote\footnote{For notational simplicity, we dot  not make explicit the dependence of the  system on $N$.} this stochastic interacting particle system by $\mathbf{V}_t = (V_t^1,\ldots,V_t^N)$. 

Hence, any pair of particles $i$ and $j$ with velocities $v=V_t^i$ and $v_*=V_t^j$ interact with deviation angle $\theta$ at rate $\beta(\theta)/2(N-1)$, and then they update their velocities to $v'$ and $v_*'$ given by \eqref{eq:pre_col},
 with  $\sigma \in \mathbb{S}^2$ uniformly chosen at random among unitary vectors  such that $\sigma \cdot (v-v_*)/|v-v_*| = \cos \theta$; notice that $v'$ and $v_*'$ as defined in \eqref{eq:pre_col} now appear in the role of \emph{post-collisional velocities}, consistently with the weak form  \eqref{eq:weak} of equation \eqref{eq:B}. The fact that the function $\beta(\theta)$ has a non-integrable singularity at $\theta = 0$ entails that particles jump infinitely many times on every finite time interval (all but finitely many jumps corresponding to very small deviation angles). One can check that  the quantities  $\sum_{i=1}^N v^i$ and $ \sum_{i=1}^N |v^i|^2$, corresponding to momentum and kinetic energy, are a.s.\ exactly preserved  by the corresponding random dynamics. 

An explicit construction of the system $\mathbf{V}$  will be given in  Lemma \ref{lemaSDE}, Section \ref{sec:construction}. Let us mention for the moment that, under the assumptions we will make, a unique (in law) Markov process  with c\`{a}dl\`{a}g trajectories and generator \eqref{eq:gen} will exist, for each value of $N\in \NN$.

The goal of the present paper is to establish a fully explicit, uniform in time propagation of chaos rate, for the  Kac $N$-particle system $\mathbf{V}$. We adopt here a probabilistic  pathwise  approach, as  pioneered by  Tanaka  \cite{tanaka1978,tanaka1979} and Sznitman \cite{sznitman1989}. The main idea is to  extend the  coupling techniques  for binary-jump particle systems introduced in \cite{cortez-fontbona2016} to the much more difficult  framework of the Boltzmann equation. We will also  rely on the analytic approach and estimates of Fournier and Mischler \cite{fournier-mischler2016}. Moreover, combining these ideas  with a  uniform-in-$N$  equilibration result,  recently established by Rousset \cite{rousset2014} for Kac particles in the Maxwell molecules case, we will obtain the sharpest propagation of chaos estimates in Wasserstein distance so far available in this setting.

\subsection{Main results}

Let us first fix some notation and specify our hypotheses. Given a metric space $E$, $p\geq 1$ and $k\in\NN$, let $\P(E)$, $\P_p(E)$ and $\P_p^\textnormal{sym}(E^k)$ respectively denote the space of probability measures on $E$, the subspace of probability measures on $E$ with finite $p$-moment, and the subspace of $\P_p(E^k)$ consisting of symmetric probability measures on $E^k$ with finite $p$-moment. Given a vector $\mathbf{x} = (x^1,\ldots,x^N) \in (\RR^3)^N$, we define the empirical measures $\bar{\mathbf{x}} \in \P(\RR^3)$ and $\bar{\mathbf{x}}^i \in \P(\RR^3)$ for any $i\in\{1,\ldots,N\}$, by
\begin{equation} \label{eq:emp_meas_x}
\bar{\mathbf{x}} := \frac{1}{N} \sum_{j=1}^N \delta_{x^j}
\quad \text{and} \quad
\bar{\mathbf{x}}^i := \frac{1}{N-1} \sum_{j\neq i} \delta_{x^j}.
\end{equation}
The empirical measure of the particle system at time $t\geq 0$ is thus denoted by $\bar{\mathbf{V}}_t$. Also, given a (mainly exchangeable) random vector $\mathbf{X}$ on $(\RR^3)^N$, we will denote its law by $\law(\mathbf{X}) \in \P((\RR^3)^N)$, and the law of its $k$ first components by $\law^k(\mathbf{X}) \in \P((\RR^3)^k)$, for any $k\leq N$.

For $\mu,\nu \in \P_2((\RR^3)^k)$, their \emph{quadratic Wasserstein distance} is defined as
\begin{align*}
\W_2(\mu,\nu)
&=  \inf_\pi \left( \int_{(\RR^3)^k \times (\RR^3)^k}
   |\mathbf{x}-\mathbf{y}|_k^2 \pi(d\mathbf{x},d\mathbf{y}) \right)^{1/2}
= \inf_{\mathbf{X},\mathbf{Y}} \left( \EE |\mathbf{X}-\mathbf{Y}|_k^2 \right)^{1/2},
\end{align*}
where the first infimum is taken over all $\pi \in \P_2((\RR^3)^k \times (\RR^3)^k)$
having marginals $\mu$ and $\nu$, and the second infimum is taken over all random
vectors $\mathbf{X}$ and $\mathbf{Y}$ on $(\RR^3)^k$ such that $\law(\mathbf{X}) = \mu$
and $\law(\mathbf{Y}) = \nu$.
Here we use the \emph{normalized} distance $|\cdot|_k$ on $(\RR^3)^k$ defined by $|\mathbf{x}|_k^2 = \frac{1}{k} \sum_{i=1}^k |x^i|^2$.
It is known that the infimum is always reached,
and a $\pi$ attaining the first infimum or a pair $(\mathbf{X},\mathbf{Y})$
attaining the second one, is referred to as an \emph{optimal coupling}.

The angular cross section function $\beta$ will be assumed to satisfy
\begin{equation}
\label{eq:nu}
\exists\nu\in(0,1), \quad c_0 \theta^{-1-\nu} \leq \beta(\theta) \leq c_1 \theta^{-1-\nu} \quad \forall \theta \in(0,\pi/2),
\end{equation}
for some constants $0<c_0<c_1$.
The initial distribution $f_0$ will be assumed to satisfy
\begin{equation}
\label{eq:p0}
\exists p_0>4, \quad \int_{\RR^3} |v|^{p_0} f_0(dv) < \infty.
\end{equation}
Note that in the usual Maxwellian  case, the quantities $\nu$ and $\gamma = 0$
are linked by the relation $\nu= \frac{1-\gamma}{2}$,
but in our (slightly more general) context, $\nu\in (0,1)$ will be viewed as an  independent parameter.

We now state our main results. See Definition \ref{def:weak} below for the notion of weak solutions of \eqref{eq:B} that we will use, and Theorem \ref{thm:well_posed} for its well-posedness .

\begin{thm} \label{thm:main}
Assume \eqref{eq:nu} and \eqref{eq:p0}, and
let $(f_t)_{t\geq 0}$ be the unique weak solution of \eqref{eq:B}.
Let $G_0^N \in \P_2^{\textnormal{sym}}((\RR^3)^N)$ be given,
and let $\mathbf{V}_t = (V_t^1,\ldots,V_t^N)$ be the particle system
with generator \eqref{eq:gen} starting with law $G_0^N$.
Then, there exists a constant $C>0$ such that for
all $t\geq 0$,
\[
\EE\W_2^2(\bar{\mathbf{V}}_t,f_t) \leq C \W_2^2(G_0^N,f_0^{\otimes N}) + C(1+t)^2 N^{-1/3}.
\]
\end{thm}

Note that one can simply choose $G_0^N = f_0^{\otimes N}$ and then $\W_2^2(G_0^N,f_0^{\otimes N})=0$,
or assume that the term $\W_2^2(G_0^N,f_0^{\otimes N})$ goes to $0$ at least  as fast as $N^{-1/3}$.  In either case, the previous theorem gives a chaos result in squared $2$-Wasserstein distance
for the Kac particle system associated to the Boltzmann equation for Maxwell molecules,
with an explicit rate of order $N^{-1/3}$. The time dependence is quadratic.

 It is worth noting that the proof of Theorem \ref{thm:main}
also provides the same convergence rates for each $k$-marginal of the the $N$-particle system defined by \eqref{eq:gen},
 when $N$ goes to $\infty$, and that statements for cutoff systems can also be established (with similar dependence on the cutoff parameter as in  \cite{fournier-mischler2016}). 

Under stronger  conditions on the initial law,  we obtain our most important result: 

\begin{thm}
\label{thm:uniform_Maxwell}
Under the same hypotheses of Theorem \ref{thm:main}, assume additionally that $\int v f_0(dv) = 0$, $\int |v|^2 f_0(dv) = 1$, and that $R_p := \sup_N \EE |V_0^1|^p < \infty$ for some $p\geq 4$.
Then, there exists $C>0$ and for all $\epsilon>0$ there exists $C_{p,\epsilon}>0$
such that for all $t\geq 0$,
\[
\EE \W_2^2(\bar{\mathbf{V}}_t, f_t)
\leq C \W_2^2(G_0^N,f_0^{\otimes N}) + C_{p,\epsilon} N^{-(p-2)/3p + \epsilon}.
\]
\end{thm}

Thus, for each such $p\geq 4 $,  we obtain a uniform propagation of chaos estimate for Maxwell molecules in squared $2$-Wasserstein distance at an explicit rate of order almost $N^{-(p-2)/3p}$, provided that $\W_2^2(G_0^N, f_0^{\otimes N})$ converges to $0$ with the same rate or faster (one can simply take $G_0^N = f_0^{\otimes N}$ as long as $f_0$ has finite $p$-moment).  For instance: if, together with \eqref{eq:p0}, one only assumes $R_p<\infty$ for $p=4$, we obtain a chaos rate of order almost $N^{-1/6}$; but if one assumes $R_p < \infty$ for all $p$, then the rate is of order almost $N^{-1/3}$.

Let us mention that our techniques   can be applied to the case in which $\nu \in [1,2)$ (implying the usual integrability condition $\int_0^{\pi/2} \theta^2 \beta(\theta) < \infty$), but this
would require  additional technicalities in order to  treat the probabilistic objects involved in the statements and proofs. To keep the exposition simple, we restrict ourselves to the case $\nu\in(0,1)$, which includes the classical Maxwellian molecules ($\nu=1/2$). 
Let us further remark that our coupling   techniques can be applied to  Kac particles systems in the  hard potentials and hard spheres cases as well,  providing in those frameworks the same rate $N^{-1/3}$ as in Theorem \ref{thm:main}, but a much worse dependence on time. This will be addressed in forthcoming works.

\subsection{Comparison to known results and approaches}\label{known results}

The study of  propagation of chaos for the Boltzmann equation was initiated in the paper
 \cite{kac1956}  by Kac. Propagation of chaos results  were proved    by several authors the  decades thereafter,  for different instances of the equation,  in  the weak convergence sense  and without convergence rates. For instance, 
 McKean \cite{mckean1967} and Gr\"{u}nbaum \cite{grunbaum1971}
obtained such type of results for some models with bounded collision kernel, whereas the work of Sznitman \cite{sznitman1984}
dealt with  unbounded kernels in the hard spheres case.
Tanaka \cite{tanaka1978,tanaka1979} introduced a crucial tool in the probabilistic study of the propagation of chaos property, the \emph{nonlinear process},  which represents
the trajectory of a ``typical particle'' in the infinite population and whose time-marginal laws solve the Boltzmann equation.  By coupling the particle system with  independent nonlinear processes,
Graham and M\'{e}l\'{e}ard \cite{graham-meleard1997}
obtained one of the first quantitative propagation of chaos results which cover cutoff variants of the model. 
 However, their approach,  relying on total
variation distance on path space, can not be extended to non-cutoff contexts, and provides  bounds which increase exponentially  in time.

More recently, in the remarkable work of Mischler and Mouhot \cite{mischler-mouhot2013}, uniform-in-time propagation of chaos results in $\W_1$ distance for 
Maxwellian molecules and hard spheres are established, with a slow   (and hard to track) rate in $N$.
Their method, of analytic nature, focused on the stability of the evolution of the
time-marginal laws of the particle system and relied on the comparison between Wasserstein and other  probability distances.  Moreover, combining Theorem 5.1-(iii) of \cite{mischler-mouhot2013} with the estimates found in Step 3 of the proof of Theorem 8-(ii) of \cite{carrapatoso2015} by Carrapatoso, one can obtain  a uniform-in-time  chaos rate  for Maxwellian molecules  of order  almost  $N^{-1/8}$ in some comparable distance, though only in the case of i.i.d.\ initial data conditioned to the Boltzmann sphere, and under a finite Fisher information assumption  (this seems to be the best uniform rate so far available in the literature;  we thank an anonymous referee for pointing out this to us).

In  \cite{fournier-mischler2016}, using a coupling with independent nonlinear processes and optimal transport based techniques,   Fournier and Mischler obtained
a propagation of chaos result with an optimal rate of order $N^{-1/2}$ in squared Wasserstein distance  $\W_2$ for \emph{Nanbu} particle systems in the hard potentials, hard spheres and Maxwell molecules cases, in  the latter setting with  the same dependence on $t$ as ours in  Theorem \ref{thm:main}.  Analogous  coupling arguments relying on optimal transport  were developed earlier  in Fontbona et al.\ \cite{fontbona-guerin-meleard2009} for Nanbu type diffusive approximations of the Landau equation, with less sharper explicit  rates (due to the general  coefficients and the suboptimal estimates for  empirical measures of i.i.d.\ samples available by that time).  Recall that, contrary to the Markov dynamics  \eqref{eq:gen} of binary (also called \emph{Bird type}) interactions, each particle  in  a Nanbu  type  system is driven by  an independent noise source,  which  implies in the jump case  that only one particle jumps at each collision.  Such particle systems  preserve momentum and energy only in mean, and  hence are  less meaningful from  the physical point view (but are still relevant for numerical simulation purposes). Since the coupling constructions in \cite{fournier-mischler2016} and \cite{fontbona-guerin-meleard2009} strongly relied on the independence of the noise sources for different particles, they cannot be  not applied  to systems with true binary interactions like \eqref{eq:gen}, where the random noises are shared.
  
In \cite{cortez-fontbona2016}, we addressed this problem in the case of Kac's one dimensional model and some generalizations. More precisely, we introduced a  new coupling between an interacting particle system with effective binary interactions and a system of nonlinear processes driven by the same  randomness  sources, which thus turned out to be  \emph{not independent}. As part of that coupling argument, we had to show, in a second step,    that these nonlinear processes become on their turn independent  as $N$ goes to $\infty$. 
With a similar strategy, the case of the Landau equation was recently addressed by  Fournier and Guillin  \cite{fournier-guillin2015} in the hard potentials and  Maxwell molecules cases. Relying  also on a stabilization result for the corresponding particle system (analogous  to the one in \cite{rousset2014}), they obtain in the latter case a uniform   propagation of chaos estimate, similar to ours in the Boltzmann case.

One of  the main  additional difficulties that pathwise probabilistic approaches to the  Boltzmann equation need to deal with, when compared to the one dimensional Kac model (and also to the Landau setting), is the  lack of continuity  of the parametrization of the collision angles, as a function of  pre-collisional velocities. Here, we  cope with this problem using optimal transport techniques and dealing  with cutoff versions of the nonlinear process.

 The assumptions on the initial distributions required in Theorem \ref{thm:uniform_Maxwell} are  similar to those  in \cite{fournier-guillin2015}, and 
are  much more general than in all the available  uniform propagation of chaos results  for Maxwell molecules (no support constraint  or regularity being needed).  We notice that the rate of $N^{-1/3}$ in  $\W_2^2$ for  Bird type particle systems obtained here,  in \cite{cortez-fontbona2016} and in \cite{fournier-guillin2015},  is slower than the $N^{-1/2}$  rate valid for Nanbu type systems (corresponding to the optimal convergence rate in expected  $\W_2^2$ distance for the empirical measure of i.i.d.\ samples, established in  \cite{fournier-guillin2013}).
An interesting  question, raised in \cite{cortez-fontbona2016}, is to what extent this sub-optimality is intrinsic  to the interaction type, or a  consequence
of the techniques employed.

\subsection{Weak solutions and nonlinear processes}

We next recall the notion of weak solutions for \eqref{eq:B} we will work with, for which some definitions are needed. We follow \cite{fournier-mischler2016}.
Consider the function $G:\RR_+ \to (0,\pi/2]$ defined as $ G(z) = H^{-1}(z)$  where $H:(0,\pi/2]\to \RR_+$ is given by
\[
H(\theta) := \int_\theta^{\pi/2} \beta(x) dx.
\]
Consider also measurable functions $\ii, \jj : \RR^3 \to \RR^3$ such that for every $x \neq 0$,
\[
\left(\frac{x}{|x|}, \frac{\ii(x)}{|x|}, \frac{\jj(x)}{|x|} \right)
\]
is an orthonormal basis of $\RR^3$.
We may and shall assume that they are  homogeneous functions, that is, one has $\ii(\lambda x) = \lambda \ii(x)$ and $\jj(\lambda x) = \lambda \jj(x)$ for all $x\in\RR^3$ and all $\lambda \in \RR$.
For $v, v_* \in \RR^3$, $\theta \in (0,\pi/2)$, $\phi \in [0,2\pi)$, and $z\geq 0$, define
\begin{equation}
\label{eq:Gamma_a_c}
\begin{split}
\Gamma(x,\phi)
&:= (\cos\phi)\ii(x) + (\sin\phi)\jj(x), \\
a(v,v_*,\theta,\phi)
&:= - \frac{1-\cos\theta}{2} (v-v_*) + \frac{\sin\theta}{2} \Gamma(v-v_*,\phi), \\
c(v,v_*,z,\phi)
&:= a(v,v_*,G(z),\phi).
\end{split}
\end{equation}
Note that when $\phi$ varies in $[0,2\pi)$, the vector $v+a(v,v_*,\theta,\phi)$ ranges all over the circle centered at $b = \frac{v+v_*}{2} + \cos\theta \frac{v-v_*}{2}$, with raduis $r = \sin\theta \frac{|v-v_*|}{2}$ and orthogonal to $d=\frac{v-v_*}{|v-v_*|}$. Denote this circle by $\C(v,v_*,\theta)$, or alternatively by $\hat{\C}(b,r,d)$.

These objects provide a suitable parametrization of the post-collisional velocities:
it is straightforward to verify that for all $v,v_*\in\RR^3$ and for any Lipschitz
bounded measurable function $\Phi$ on $\RR^3$,
\begin{align}
\label{eq:int_B}
\begin{split}
& \int_{\mathbb{S}^2} d\sigma [\Phi(v')-\Phi(v)] B(\theta) \\
&=  \int_0^{\pi/2} d\theta \int_0^{2\pi} d\phi
  [\Phi(v+a(v,v_*,\theta,\phi)) - \Phi(v)] \beta(\theta) \\
&= \int_0^\infty dz \int_0^{2\pi} d\phi [\Phi(v+c(v,v_*,z,\phi)) - \Phi(v)],
\end{split}
\end{align}
the integral being well defined 
since $|a(v,v_*,\theta,\phi)|\leq C \theta |v-v_*|$ and  $\int_0^{\pi/2}\theta \beta(\theta)d\theta<\infty$  when \eqref{eq:nu} holds.  By a slight abuse of notation, we still call $v' = v + a(v,v_*,\theta,\phi)$,
$v_*' = v_* - a(v,v_*,\theta,\phi)$ and also $v' = v + c(v,v_*,z,\phi)$, $v_*' = v_* - c(v,v_*,z,\phi)$.

\begin{defi} \label{def:weak}
We say that a collection $(f_t)_{t\geq 0} \in C([0,\infty),\P_2(\RR^3))$ is a \emph{weak solution}
for \eqref{eq:B} if it preserves momentum and energy (that is, $\int v f_t(dv) = \int v f_0(dv)$
and $\int |v|^2 f_t(dv) = \int |v|^2 f_0(dv)$ for all $t \geq 0$), and if for all
bounded Lipschitz function $\Phi : \RR^3 \to \RR$ and for all $t\geq 0$,
\begin{equation} \label{eq:weak}
\begin{split}
\int_{\RR^3} \Phi(v) f_t(dv)
&= \int_{\RR^3} \Phi(v) f_0(dv)
 + \int_0^t ds \int_{\RR^3}\int_{\RR^3} f_s(dv)f_s(dv_*)
   \int_0^{\pi/2} d\theta \beta(\theta) \\
& \qquad  \qquad \qquad \qquad \qquad {} \times    \int_0^{2\pi} d\phi [\Phi(v+a(v,v_*,\theta,\phi)) - \Phi(v)].
\end{split}
\end{equation}
\end{defi}

The next statement
provides the main analytical properties of equation \eqref{eq:B} that we shall need.
The proof of well-posedness can be found for instance in
\cite{toscani-villani1999} and \cite{villani2002},
whereas the proof of the existence of a density can be found in \cite{fournier2015}.
%See Definition \ref{def:weak}  below for the notion of weak solutions of \eqref{eq:B} we will use.

\begin{thm} \label{thm:well_posed}
Assume \eqref{eq:nu} and \eqref{eq:p0}. Then, there exists
a unique weak solution $(f_t)_{t\geq 0} \in C([0,\infty),\P_2(\RR^3))$ 
of \eqref{eq:B}, which satisfies $\sup_t \int_{\RR^3} |v|^{p_0} f_t(dv) < \infty$.
Moreover, if $f_0$ is not a Dirac mass, then $f_t$ has a density as soon as $t>0$.
\end{thm}

The  nonlinear process, introduced by Tanaka in \cite{tanaka1978,tanaka1979} to
provide a probabilistic interpretation of the Boltzmann equation,  can be defined in the present case through a stochastic integral equation
with respect to some Poisson point measure.
More specifically, consider the equation
\begin{equation}
\label{eq:nonlinear}
dW_t = \int_0^\infty \int_0^{2\pi} \int_{\RR^3} c(W_{t^-},v,z,\phi) \M(dt,dz,d\phi,dv),
\end{equation}
where $\M(dt,dz,d\phi,dv)$ is a Poisson point measure on $[0,\infty)\times[0,\infty)\times[0,2\pi)\times \RR^3$
with intensity $dtdz d\phi f_t(dv) / 2\pi$.
Following Tanaka's ideas,  under \eqref{eq:nu} and \eqref{eq:p0} Fournier and M\'{e}l\'{e}ard \cite{fournier-meleard2002}  proved weak existence and uniqueness in law for equation \eqref{eq:nonlinear}, together with the fact that $(\law(W_t))_{t\geq 0}$ solves \eqref{eq:weak}, hence $\law(W_t) = f_t$ for all $t$. Any process having the same law as $W$ is called a nonlinear process.

Unfortunately,  we cannot  carry out our coupling construction by  relying  only on weak  existence of solutions to equation  \eqref{eq:nonlinear}. This is why we will need to work with  the  \emph{cutoff}
nonlinear process instead. Given a cutoff level $L>0$,  this process  $W^L$
can be  defined as the solution of a nonlinear SDE similar  to \eqref{eq:nonlinear},
namely 
\begin{equation}
\label{eq:nonlinear_cutoff}
dW_t^L = \int_0^\infty \int_0^{2\pi} \int_{\RR^3} c_L(W_{t^-}^L,v,z,\phi) \M^L(dt,dz,d\phi,dv),
\end{equation}
where $c_L(v,v_*,z,\phi) := c(v,v_*,z,\phi) \ind_{\{z\leq L\}}$.
This time $\M^L(dt,dz,d\phi,dv)$ is a Poisson point measure
on $[0,\infty)\times[0,\infty)\times[0,2\pi)\times \RR^3$
with intensity $dtdz d\phi f_t^L(dv) / 2\pi$,
where $(f_t^L)_{t\geq 0} \in C([0,\infty),\P_2(\RR^3))$ 
is the unique solution to the nonlinear equation \eqref{eq:weak} with  $\beta$ replaced by $\beta_L(z) := \beta(z) \ind_{\{\theta\geq G(L)\}}$, for which the well-posedness part of Theorem \ref{thm:well_posed} applies.
Strong well-posedness for \eqref{eq:nonlinear_cutoff} is straightforward: thanks to the indicator $\ind_{\{z\leq L\}}$, the equation is nothing but a recursion for the values of $W_t^L$ at its timely ordered jump-times.
By standard arguments, it can be seen that any (weak or strong) solution to \eqref{eq:nonlinear_cutoff}
satisfies $\law(W_t^L) = f_t^L$
(more specifically, the collection$(\law(W_t^L))_{t\geq 0}$ satisfies a linearized and cutoff version of \eqref{eq:weak},
which in turn has $(f_t^L)_{t\geq 0}$ as the unique solution, see for instance
Theorem 3.1 of \cite{fournier-meleard2002} for details).
As expected, one can show that $f_t^L \to f_t$ as $L\to\infty$ (see Lemma \ref{lem:ftL} below).

 \subsection{Idea of the  proofs and plan of the paper}

To prove our results,   following  ideas introduced in \cite{cortez-fontbona2016},  for each $N$ and cutoff parameter $L$ we will first couple \emph{in some  optimal way},  a suitable realization  of the particle system $\mathbf{V}_t$  with generator \eqref{eq:gen} 
(given  below in  \eqref{eq:PS})  with some system $\mathbf{U}_t^L = (U_t^{1,L},\ldots,U_t^{N,L})$
of copies of the cutoff nonlinear process $W^L$.  To do this, we will make use of optimal transport theory,
in order to carefully construct  the jumps of the system $\mathbf{U}^L$, in such a way
that they mimic as closely as possible the jumps of the particle system $\mathbf{V}$. Roughly speaking, from this construction and Gronwall's lemma, we will obtain an estimate like
\[
\EE \frac{1}{N} \sum_{i=1}^N |V_t^i - U_t^{i,L}|^2
\leq C \left[ \W_2^2(G_0^N, f_0^{\otimes N}) + (1+t)^2 \EE \W_2^2(\bar{\mathbf{U}}_t^L, f_t^L) + tL^{1-2/\nu} \right]
\]
for some constant $C>0$. As mentioned earlier, the fact that we deal with a particle
system with effective binary collisions
will imply that the cutoff nonlinear processes $U^{1,L},\ldots,U^{N,L}$ thus constructed
are \emph{not independent}. Therefore, in a similar way as in \cite{cortez-fontbona2016}, we will need
to ``decouple'', in a second step,  these processes, obtaining
\[
\EE \W_2^2(\bar{\mathbf{U}}_t^L, f_t^L)
\leq C N^{-1/3}
\]
uniformly on $L$ (see Lemma \ref{lem:decoup} and Corollary \ref{cor:EEW2overU}). We will then make $L\to \infty$ to deduce the estimates of Theorem \ref{thm:main}. Finally, from this and the results of \cite{rousset2014}, we will obtain the uniform-in-time chaos rate stated in Theorem \ref{thm:uniform_Maxwell}.

In Section \ref{sec:construction} we give an explicit construction
of the particle system $\mathbf{V}_t$ and, more importantly, 
we construct  the coupling with the corresponding system $\mathbf{U}_t^L$
of non-independent cutoff nonlinear processes  used throughout this paper.
In Section \ref{sec:estimates} we state and prove several technical
results. In Section \ref{sec:proof_main} we prove Theorem \ref{thm:main}.
Finally, in Section \ref{sec:unif_chaos} we prove Theorem \ref{thm:uniform_Maxwell},
along with some intermediate results that have interest on their own, such as the propagation of moments for the particle system without assuming bounded initial energy (see Corollary \ref{cor:unif_moments}),
and an equilibration result that extends the one by Rousset \cite{rousset2014} (see Lemma \ref{lem:rousset}).

%%%%%%%%%%%%%%%%%%%%%%%%%%%%%%%%%%%%%%%%%%%%%%%%%%%%%%%%
%%%%%%%%%%%%%%%%%%%%%%%%%%%%%%%%%%%%%%%%%%%%%%%%%%%%%%%%
\section{Construction} \label{sec:construction}
%%%%%%%%%%%%%%%%%%%%%%%%%%%%%%%%%%%%%%%%%%%%%%%%%%%%%%%%
%%%%%%%%%%%%%%%%%%%%%%%%%%%%%%%%%%%%%%%%%%%%%%%%%%%%%%%%

In this section we explicitly construct
the coupled system $(\mathbf{V},\mathbf{U}^L)$
  in order to prove our results.
These processes will be defined as solutions of stochastic integral equations driven
by the same Poisson point measure.
We follow \cite{cortez-fontbona2016} and  \cite{fournier-mischler2016}.

%%%%%%%%%%%%%%%%%%%%%%%%%%%%%%%%%%%%%%%%%%%
\subsection{The particle system}\label{subsect:The particle system}
%%%%%%%%%%%%%%%%%%%%%%%%%%%%%%%%%%%%%%%%%%%

Fix the number of particles $N\in\NN$.
We introduce the function $\i:[0,N)\to\{1,\ldots,N\}$ by $\i(\xi) = \lfloor \xi \rfloor +1$, so that $\i(\xi)$ is a discrete index associated to the continuous variable $\xi$. Let $\G\subseteq\RR^2$ be the set
\[
\G = \{(\xi,\zeta) \in [0,N)^2: \i(\xi)\neq \i(\zeta) \}.
\]
Note that its area is $|\G|= N(N-1)$. Consider now a Poisson point measure $\N(dt,dz,d\phi,d\xi,d\zeta)$ on $[0,\infty)\times [0,\infty)\times[0,2\pi)\times [0,N)\times [0,N)$ with intensity
\[
\frac{N}{2}dt dz \frac{d\phi}{2\pi} \frac{d\xi d\zeta \ind_\G(\xi,\zeta)}{|\G|}
= \frac{dt dz d\phi d\xi d\zeta \ind_{\G}(\xi,\zeta)}{4(N-1)\pi}.
\]
In words, the measure $\N$ picks atoms $(t,z)\in[0,\infty)^2$ with  intensity $\frac{N}{2} dt  dz$ and for each such atom it also independently samples an angle $\phi$ uniformly from $[0,2\pi)$ and a pair $(\xi,\zeta)$ uniformly from the set $\G$. We will use the variables $\xi$ and $\zeta$ to choose indexes $i=\i(\xi)$ and $j=\i(\zeta)$ of the particles that interact at each jump. Additionally, given $G_0^N\in\P_2^{\textnormal{sym}}((\RR^3)^N)$ and $f_0\in\P_2(\RR^3)$ as in the statement of Theorem \ref{thm:main}, we will in the sequel denote  by 
\begin{equation}\label{V0U0}
(\mathbf{V}_0,\mathbf{U}_0)
\end{equation}
 a realization, independent of $\N$, of the optimal coupling between $G_0^N$ and $f_0^{\otimes N}$. Call $\F = (\F_t)_{t\geq 0}$ the complete, right-continuous filtration generated by $(\mathbf{V}_0,\mathbf{U}_0)$ and $\N$. We denote by  $\PP$ and $\EE$ the probability measure and expectation in the corresponding probability space.

We can now introduce the particle system $\mathbf{V}  = (V^1,\ldots,V^N)$ as the solution, starting from the initial condition $\mathbf{V}_0$, of the stochastic equation
\begin{equation}
\label{eq:PS}
d\mathbf{V}_t
= \int_0^\infty \int_0^{2\pi} \int_{[0,N)^2} \sum_{i\neq j} \ind_{\{\i(\xi) = i,\i(\zeta)=j\}}
 \mathbf{c}_{ij}(\mathbf{V}_{t^-},z,\phi) \N(dt,dz,d\phi,d\xi,d\zeta)
\end{equation}
where $\mathbf{c}_{ij}(\mathbf{x},z,\phi) \in (\RR^3)^N$ is the vector with coordinates given by
\begin{equation}
\label{eq:cij}
(\mathbf{c}_{ij}(\mathbf{x},z,\phi))^l = \begin{cases}
c(x^i,x^j,z,\phi) & \text{if $l=i$}, \\
-c(x^j,x^i,z,\phi) & \text{if $l=j$}, \\
0 & \text{otherwise}.
\end{cases}
\end{equation}
Weak existence and uniqueness of solutions for \eqref{eq:PS} holds, see Lemma \ref{lemaSDE} below.

Since we required $\ii$ and $\jj$ to be homogeneous functions, in particular they are odd functions and it can be easily seen that $c(v,v_*,z,\phi) = - c(v_*,v,z,\phi)$.
Given a solution to  \eqref{eq:PS}, it follows that for each $i=1,\ldots,N$, the particle $V^i$ satisfies the stochastic equation \begin{equation} \label{eq:Vi}
dV_t^i = \int_0^\infty \int_0^{2\pi} \int_0^N c(V_{t^-}^i,V_{t^-}^{\i(\xi)}, z,\phi) \N^i(dt,dz,d\phi,d\xi),
\end{equation}
where $\N^i$ is given by
\[
\N^i(dt,dz,d\phi,d\xi)
:= \N(dt,dz,d\phi,d\xi,[i-1,i)) + \N(dt,dz,d\phi,[i-1,i),d\xi).
\]
That is, $\N^i$ selects only the atoms of $\N$ such that either $\i(\xi) = i$ or $\i(\zeta) = i$. Clearly, $\N^i$ is a Poisson point measure on $[0,\infty)\times[0,\infty)\times[0,2\pi)\times[0,N)$ with intensity
\[
\frac{dt dz d\phi d\xi \ind_{A^i}(\xi)}{2(N-1)\pi},
\]
where $A^i := [0,N)\setminus [i-1,i)$. Thus, the term $V_{t^-}^{\i(\xi)}$ appearing in \eqref{eq:Vi} is a $\xi$-realization of the (random) probability measure $\bar{\mathbf{V}}_{t^-}^i$ (defined as in \eqref{eq:emp_meas_x}). Therefore, from the point of view of the particle $V^i$, the dynamics is as follows: $(t,z)$-atoms are sampled with intensity 1 and for each such atom an angle $\phi$ is chosen and a particle $v_* = V_{t^-}^{\i(\xi)}$ is selected at random among all the others;  $v=V_{t^-}^i$  (and $v_*=V_{t^-}^{\i(\xi)}$) then updates its  state to $v'=V_{t}^i$  (and $v_*'=V_{t}^{\i(\xi)}$) as given in  \eqref{eq:pre_col}.

%%%%%%%%%%%%%%%%%%%%%%%%%%%%%%%%%%%%%%%%%%%
\subsection{Coupling with  a system of cutoff nonlinear processes}
%%%%%%%%%%%%%%%%%%%%%%%%%%%%%%%%%%%%%%%%%%%

The key observation is the following: in \eqref{eq:Vi} with $c_L$ in place of $c$, if one replaces $V_{t^-}^{\i(\xi)}$
by some realization of the probability measure $f_t^L$,
then  the resulting equation defines a cutoff nonlinear process as in \eqref{eq:nonlinear_cutoff}.
Moreover: we want to choose this $f_t^L$-distributed random variable in an optimal
way (in the suitable sense), in order that the resulting process remains close to $V^i$.
Such a construction needs to be carried out in a measurable way, which motivates the following lemma. In the sequel, all optimal couplings and optimal costs considered use the cost function $C(v,u) = |v-u|^2$.

\begin{lem}[coupling] \label{lem:coup}
Fix $L>0$ and $i\in\{1,\ldots,N\}$. Then, there exists an $\RR^3$-valued function $\Pi_t^{i,L}(\mathbf{x},\xi)$,
measurable in $(t,\mathbf{x},\xi) \in [0,\infty)\times(\RR^3)^N\times A^i$,
with the following property: for any $(t,\mathbf{x})\in[0,\infty)\times(\RR^3)^N$
and any random variable $\xi$ uniformly chosen in $A^i$,
the pair $(x^{\i(\xi)},\Pi_t^{i,L}(\mathbf{x},\xi))$ is an optimal coupling
between $\bar{\mathbf{x}}^i$ and $f_t^L$.
Moreover, for any exchangeable random vector $\mathbf{X}\in(\RR^3)^N$ and any bounded  measurable
function $h$, we have $\EE \int_{j-1}^j h(\Pi_t^{i,L}(\mathbf{X},\xi)) d\xi = \int_{\RR^3} h(u) f_t^L(du)$
for any $j\in\{1,\ldots,N\}$, $j\neq i$.
\end{lem}

\begin{proof}
See  Lemma 3 in \cite{cortez-fontbona2016}.
\end{proof}

To ensure that the post-collisional velocities of $V^i_t$ and $U^{i,L}_t$ do not differ much,
we will use the functions $\Pi^{i,L}$ of Lemma \ref{lem:coup}
to define our system $\mathbf{U}^L=(U^{1,L},\ldots,U^{N,L})$ of cutoff nonlinear processes.  This will mean that at each jump of $v=V_t^i$ together with some other particle $v_*$,
the corresponding process $u=U_t^{i,L}$ will sample some $f_t^L$-distributed variable $u_*$
to interact with, in such a way that the interactions of the system 
$\mathbf{U}^L$ mimic those of the particle system $\mathbf{V}$.

However, post-collisional  velocities will  also depend on the angles $\phi$ chosen in circles of the form $\C(v,v_*,\theta)$ and $\C(u,u_*,\theta)$  associated with each  collision.
As  remarked by Tanaka  \cite{tanaka1978,tanaka1979},  no continuity assumption can be made about the functions $\ii$ and $\jj$, and,  in order to control the distance between  $V_t^i$ and $U^{i,L}_t$ after a collision, one has to make specific (non trivial) uniformly random choices for those angles as well. In the present paper, we will choose the angles $\phi$  uniformly in the circles  $\C(v,v_*,\theta)$ and $\C(u,u_*,\theta)$ in  such a way that their joint distribution is  an optimal coupling of the uniform laws on these circles, with respect to the quadratic cost. 
The optimal transport cost happens to have a nice explicit formula, with the
optimal transport map depending only on $v-v_*$ and $u-u_*$, in a fully explicit way. This is
stated in the following:

\begin{lem}[optimal coupling of circles] \label{lem:optimal_circles}
Recall that $\hat{\C}(b,r,d)$ denotes the circle centered at $b\in\RR^3$, with radius $r>0$ and orthogonal
direction $d\in\SS^2$; alternatively, for $v,v_*\in\RR^3$ and $\theta\in[0,\pi/2]$,
$\C(v,v_*,\theta)$ denotes the circle centered at
$\frac{v+v_*}{2} + \cos\theta \frac{v-v_*}{2}$, with radius $\sin\theta \frac{|v-v_*|}{2}$
and orthogonal to $\frac{v-v_*}{|v-v_*|}$. Then:

\begin{enumerate}

\item \label{lem:W2circles}
For any $b,\tilde{b}\in\RR^3$, $r,\tilde{r}\geq 0$ and $d,\tilde{d} \in \SS^2$, the optimal transport cost between the uniform distributions on the circles $\hat{\C}(b,r,d)$ and $\hat{\C}(\tilde{b},\tilde{r},\tilde{d})$ is given by
\begin{equation} \label{eq:W2circles}
\W_2^2\left(\textnormal{unif}_{\hat\C(b,r,d)}, \textnormal{unif}_{\hat\C(\tilde{b},\tilde{r},\tilde{d})} \right)
= |b-\tilde{b}|^2 + (r-\tilde{r})^2 + r\tilde{r}(1- |d\cdot\tilde{d}|) .
\end{equation}

\item \label{lem:varphi}
There exists a measurable function $\varphi : \RR^3 \times \RR^3 \times [0,2\pi) \to [0,2\pi)$ with the following property: for every $v,v_*,u,u_* \in \RR^3$, $\theta,\vartheta \in[0,\pi/2]$ and for any random variable $\phi$ uniformly chosen in $[0,2\pi)$, the pair
\[
\left( v + a(v,v_*,\theta,\phi), u + a(u,u_*,\vartheta,\varphi) \right)
\]
where $\varphi = \varphi(v-v_*,u-u_*,\phi)$, is an optimal coupling of the uniform distributions on the circles $\C(v,v_*,\theta)$ and $\C(u,u_*,\vartheta)$.

\end{enumerate}
\end{lem}

\begin{proof}
We first prove \ref{lem:W2circles}.
Without loss of generality, assume $\tilde{b}=0$ and $d\cdot\tilde{d}\geq 0$. Let $h =h(d,\tilde d)\in \SS^2$ be a  fixed measurable choice of a unitary vector orthogonal to both $d$ and $\tilde{d}$ (if they are parallel then there are infinitely many such $h$'s; if not, there are only $2$; we can take for instance  $h=\ii(d)/|d|$ in the first case and $h=d\times \tilde d /|d\times \tilde d  |$ in the second). 
Let also  $k,\tilde{k}\in\SS^2$ be such that $(h,k,d)$ and $(h,\tilde{k},\tilde{d})$ are orthonormal bases of $\RR^3$ with the same orientation $s=d\cdot (h\times k)=\tilde d\cdot  (h\times \tilde k) \in \{-1,1\}$,  
 so that $k\cdot\tilde{k}=s k \cdot ( \tilde{d}\times h) =s\tilde{d}\cdot (h\times k)=d\cdot  \tilde{d}$. With these bases, we can now parametrize the circles $\hat\C(b,r,d)$ and $\hat \C(\tilde{b},\tilde{r},\tilde{d})$ using angles $\phi$ and $\tilde{\phi}\in[0,2\pi)$. Namely, a point $x\in \hat\C(b,r,d)$ is written as $x=b+r(\cos\phi)k+r(\sin\phi)h$, while a point $y\in \hat\C(\tilde{b},\tilde{r},\tilde{d})$ is written as $y=\tilde{r}(\cos\tilde{\phi})\tilde{k}+\tilde{r}(\sin\tilde{\phi})h$. Then, the associated cost is
\begin{align*}
C(\phi,\tilde{\phi})
%= |x-y|^2
&= |b+r(\cos\phi)k - \tilde{r}(\cos\tilde{\phi})\tilde{k} + (r\sin\phi - \tilde{r}\sin\tilde{\phi})h |^2 \\
%&= |b|^2 + r^2(\cos\phi)^2 + \tilde{r}^2(\cos\tilde{\phi})^2 +	(r\sin\phi - \tilde{r}\sin\tilde{\phi})^2 \\
%&\quad {} + 2b\cdot[r(\cos\phi)k - \tilde{r}(\cos\tilde{\phi})\tilde{k}
%               + (r\sin\phi - \tilde{r}\sin\tilde{\phi})h] \\
%&\quad {} - 2 r\tilde{r}(\cos\phi\cos\tilde{\phi}) k\cdot\tilde{k} \\
&= |b|^2 + r^2 + \tilde{r}^2
   - 2 r\tilde{r}[(\sin\phi\sin\tilde{\phi}) + (\cos\phi\cos\tilde{\phi}) d\cdot\tilde{d}] \\
&\quad {} + 2b\cdot[r(\cos\phi)k - \tilde{r}(\cos\tilde{\phi})\tilde{k}
               + (r\sin\phi - \tilde{r}\sin\tilde{\phi})h].
\end{align*}
Using the inequality $2\alpha\beta \leq \alpha^2 + \beta^2$ in the cross-terms, we obtain $C(\phi,\tilde{\phi}) \geq \Phi(\phi) - \Psi(\tilde{\phi})$ for all $\phi,\tilde{\phi} \in[0,2\pi)$, where
\begin{align*}
\Phi(\phi)
&= |b|^2 + r^2 
   - r\tilde{r}[(\sin\phi)^2 + (\cos\phi)^2 d\cdot\tilde{d}]
   + 2r b\cdot[(\cos\phi) k + (\sin\phi) h] \\
\Psi(\tilde{\phi})
&= -\{ \tilde{r}^2 
   - r\tilde{r}[(\sin\tilde{\phi})^2 + (\cos\tilde{\phi})^2 d\cdot\tilde{d}]
   - 2\tilde{r} b\cdot[(\cos\tilde{\phi}) \tilde{k} + (\sin\tilde{\phi}) h] \}.
\end{align*}
Moreover, the equality $C(\phi,\tilde{\phi}) = \Phi(\phi) - \Psi(\tilde{\phi})$ is attained when $\phi=\tilde{\phi}$. Using for instance Remark 5.13 of \cite{villani2009}, this shows that taking $\phi = \tilde{\phi}$ uniformly distributed on $[0,2\pi)$ in fact provides an optimal coupling of the uniform distributions on $\C(b,r,d)$ and $\C(\tilde{b},\tilde{r},\tilde{d})$. This proves point \ref{lem:W2circles}, since the cost of this coupling is
\begin{align*}
\int_0^{2\pi} C(\phi,\phi) \frac{d\phi}{2\pi}
&= \frac{1}{2\pi} \int_0^{2\pi}
   \{|b|^2 + r^2 + \tilde{r}^2
   - 2 r\tilde{r}[(\sin\phi)^2 + (\cos\phi)^2 d\cdot\tilde{d}] \\
&\qquad \qquad \quad {} + 2b\cdot[(\cos\phi)(rk -\tilde{r}\tilde{k}) 
               + (\sin\phi) (r - \tilde{r})h]\} d\phi \\
&= |b|^2 + r^2 + \tilde{r}^2 - r\tilde{r}(1+d\cdot \tilde{d}).
%&= |b|^2 + (r - \tilde{r})^2 + r\tilde{r}(1-d\cdot \tilde{d}). 
\end{align*}

We now prove \ref{lem:varphi}. Put $d = \frac{v-v_*}{|v-v_*|}$ and $\tilde{d} = \frac{u-u_*}{|u-u_*|}$. For some fixed measurable choice  $(d,\tilde{d})\mapsto h=h(d,\tilde{d})$ of a vector  $h$ orthogonal to both $d$ and $\tilde{d}$,  let $\phi_i=\phi_i(d,\tilde{d}) \in[0,2\pi)$, $i=1,2$   be the unique angles such that 
\[
\frac{\Gamma(v-v_*,\phi_1)}{|v-v_*|} = \frac{\Gamma(u-u_*,\phi_2)}{|u-u_*|} =: h,
\]
 Note that $(\phi_1,\phi_2)$ depend only on $v-v_*$ and $u-u_*$ through $d$, $\tilde{d}$, $\Gamma(v-v_*,\cdot)$ and $\Gamma(u-u_*,\cdot)$, in a measurable way. Now put
\[
k := \frac{\Gamma(v-v_*,\phi_1 + \pi/2)}{|v-v_*|},
\qquad \tilde{k} := \frac{\Gamma(u-u_*,\phi_2+s\pi/2)}{|u-u_*|}.
\]
Here $s = \pm 1$ is chosen such that the rotation in $\pi/2$ is performed
with the same orientation.
More specifically, if $d\cdot \tilde{d} \geq 0$, then
$s=1$ when the bases $(\frac{v-v_*}{|v-v_*|},\frac{\ii(v-v_*)}{|v-v_*|},\frac{\jj(v-v_*)}{|v-v_*|})$
and $(\frac{u-u_*}{|u-u_*|},\frac{\ii(u-u_*)}{|u-u_*|},\frac{\jj(u-u_*)}{|u-u_*|})$
have the same orientation, and $s=-1$ otherwise; but when
$d\cdot \tilde{d} <0$, we make the opposite choice.
Now, the same argument of part \ref{lem:W2circles} shows that if $\phi$ is a uniform random variable on $[0,2\pi)$ then $v+ a(v,v_*,\theta,\phi)$ and $u + a(u,u_*,\vartheta,s(\phi-\phi_1)+\phi_2)$ constitute an optimal coupling. Put $\varphi = s(\phi - \phi_1)+\phi_2$ and the conclusion follows.
\end{proof}

\begin{rmk}
\begin{itemize}
\item The expression on the right of \eqref{eq:W2circles} is nice: the term $|b-\tilde{b}|^2$ is the cost associated to translation of the circles, the term $(r-\tilde{r})^2$ is the dilation or contraction cost, and $r\tilde{r}(1-|d\cdot\tilde{d}|)$ corresponds to inclination.
\item When $(v-v_*)\cdot (u-u_*) \geq 0$ and $\theta=\vartheta$,
the coupling given in Lemma \ref{lem:optimal_circles}--\ref{lem:varphi}
reduces to the parallel spherical coupling of \cite{rousset2014}.
\end{itemize}
\end{rmk}

With the functions $\Pi^{i,L}$ and $\varphi$ of Lemmas \ref{lem:coup}
and \ref{lem:optimal_circles}--\ref{lem:varphi} in hand,
we can now introduce, at a formal level first,  
a system of cutoff nonlinear processes $\mathbf{U}^{L} = (U^{1,L},\ldots,U^{N,L})$, suitably  constructed in the same probability space as  $\mathbf{V}  = (V^1,\ldots,V^N)$ defined in \eqref{eq:PS}.
Recall that  the pair $(\mathbf{V}_0,\mathbf{U}_0)
$  is given and specified in  \eqref{V0U0}.
Mimicking \eqref{eq:Vi}, for each $L\in[1,\infty)$, $N\in \NN$ and $i=1,\dots N$,  the processes  $U^{i,L}$ is defined as the solution, starting from $U_0^i$,
of the stochastic equation
\begin{equation} \label{eq:Ui}
dU_t^{i,L} = \int_0^\infty \int_0^{2\pi} \int_0^N c_L(U_{t^-}^{i,L},\Pi_t^{i,L}(\mathbf{U}_{t^-}^L,\xi), z,\varphi_{t^-}^i) \N^i(dt,dz,d\phi,d\xi),
\end{equation}
where we have used the shorthand
\begin{equation}
\label{eq:varphi_t}
\varphi_{t^-}^i = \varphi(V_{t^-}^i - V_{t^-}^{\i(\xi)}, U_{t^-}^{i,L} - \Pi_t^{i,L}(\mathbf{U}_{t^-}^L,\xi), \phi ).
\end{equation}
In words, at  jump instants of $U_t^{i,L}$, this process collides with an   $f_t^L$-distributed random variable,
which is optimally coupled to the realization $U_{t^-}^{\i(\xi),L}$ of the (random) measure $\bar{\mathbf{U}}_{t^-}^{i,L}$.
Since the Poisson measures $\N^i$ and $\N^j$ share some of its atoms,
processes $U^{i,L}$ and $U^{j,L}$ have simultaneous jumps and hence are not independent.

We can  write the following joint SDE for the pair $(\mathbf{V},\mathbf{U}^L)$, arranged
as the collection of pairs $((V^1,U^{1,L}), \ldots, (V^N,U^{N,L})) \in (\RR^3\times\RR^3)^N$:
\begin{equation}\label{SDEpartnonlin}
\begin{split}
d(\mathbf{V},\mathbf{U}^L)_t
&= \int_0^\infty \int_0^{2\pi} \int_{[0,N)^2} \sum_{i\neq j} \ind_{\{\i(\xi) = i,\i(\zeta)=j\}} \\
&\qquad\qquad\qquad {} \times \mathbf{b}_{L,ij}(\mathbf{V}_{t^-},\mathbf{U}_{t^-}^L,t, z,\phi,\xi,\zeta) \N(dt,dz,d\phi,d\xi,d\zeta),
\end{split}
\end{equation}
where $\mathbf{b}_{L,ij}(\mathbf{x},\mathbf{y},t, z,\phi,\xi,\zeta) \in (\RR^3\times\RR^3)^N$ is given by
\begin{align*}
& (\mathbf{b}_{L,ij}(\mathbf{x},\mathbf{y},t, z,\phi,\xi,\zeta))^\ell \\
&= \begin{cases}
\left(\vphantom{\hat{M}}c(x^i,x^j,z,\phi) ~,~
c_L(y^i,\Pi_t^{i,L}(\mathbf{y},\zeta),z,\varphi(x^i-x^j,y^i-\Pi_t^{i,L}(\mathbf{y},\zeta),\phi)) \right) 
& \text{if $\ell=i$}, \\
\left(\vphantom{\hat{M}}c(x^j,x^i,z,\phi) ~,~
c_L(y^j,\Pi_t^{j,L}(\mathbf{y},\xi),z,\varphi(x^j-x^i,y^j-\Pi_t^{j,L}(\mathbf{y},\xi),\phi)) \right) 
& \text{if $\ell=j$}, \\
(0 ~,~ 0) & \text{otherwise}.
\end{cases}
\end{align*}

Existence for each $L\in[1,\infty)$ and $N\in \NN$   of a pair $(\mathbf{V},\mathbf{U}^L)$  solving  \eqref{SDEpartnonlin},  along with its relevant properties, is stated in the next result.
Some  arguments of the proof are standard or can be adapted from previous works, so  details will be provided only when needed. The proof is given in the Appendix.

\begin{lem}\label{lemaSDE}
Assume \eqref{eq:nu} and let $L\in[1,\infty)$ and $N\in \NN$. We have:
\begin{enumerate}
\item \label{lemaSDE_i} If $\mathbf{V}$ is a solution to \eqref{eq:PS}, then it has the law of the unique $(\RR^3)^N$  valued  Markov process  with  generator given by \eqref{eq:gen}.  In particular, we  almost surely have 
  $\sum_{i=1}^{N} V_t^i= \sum_{i=1}^{N} V_0^i$ and  $\sum_{i=1}^{N} |V_t^i|^2= \sum_{i=1}^{N} |V_0^i|^2$ for all $t\geq 0$. 
  \item \label{lemaSDE_ii} There is weak existence and uniqueness  of a  solution   $(\mathbf{V},\mathbf{U}^L  )$ to the system of SDEs \eqref{SDEpartnonlin}. 
  \item \label{lemaSDE_iii} For each $i=1,\dots,N$, the process $U^{i,L}$ is a cutoff nonlinear process, and in particular we have 
  ${\cal L}(U^{i,L}_t)=f_t^L.$
  \item \label{lemaSDE_iv} Last, the collection of pairs of processes  $ (V^1,U^{1,L}), \ldots,(V^N,U^{N,L})$ is exchangeable. 
\end{enumerate}
\end{lem}

\begin{rmk}
\label{rmk:noncutoff_nonlinear} 
One can also check, using  the preservation of moments of the processes $\mathbf{U}^{L},L\in [1,\infty)$ and  Lemma  \ref{lem:ftL}, that the family of laws of $(\mathbf{V},\mathbf{U}^L  )$ has accumulation points as  $L\to \infty$ which are couplings of the particle system  \eqref{eq:PS} and a system of $N$ non-independent nonlinear processes. 
Unfortunately, due to the lack of continuity  of the functions $\varphi$ and $\Pi_t^{i,L}$,
this does not readily ensure (weak) well-posedness for the system \eqref{SDEpartnonlin} in the case $L=\infty$, which would  simplify the construction and proofs.
This is the reason why we are constrained to work with a system of cutoff nonlinear processes.
\end{rmk}

%%%%%%%%%%%%%%%%%%%%%%%%%%%%%%%%%%%%%%%%%%%%%%%%%%%%%%%%
%%%%%%%%%%%%%%%%%%%%%%%%%%%%%%%%%%%%%%%%%%%%%%%%%%%%%%%%
\section{Estimates and technical results} \label{sec:estimates}
%%%%%%%%%%%%%%%%%%%%%%%%%%%%%%%%%%%%%%%%%%%%%%%%%%%%%%%%
%%%%%%%%%%%%%%%%%%%%%%%%%%%%%%%%%%%%%%%%%%%%%%%%%%%%%%%%

We will use the following bounds a couple of times: under \eqref{eq:nu}, it can be easily seen that for some constants $0<c_2<c_3$ we have:
\begin{equation}
\label{eq:boundG}
c_2 (1+z)^{-1/\nu} \leq G(z) \leq c_3(1+z)^{-1/\nu} \quad \forall z>0.
\end{equation}

The following lemmas provide useful estimates for our purposes.
Typically, one wants to use these lemmas with $v$ and $v_*$
taken from the particle system, and $u$ and $u_*$ taken from the
system of cutoff nonlinear processes.

\begin{lem} \label{lem:int-dphi}
Write $R(v,u) := |v||u| + |v\cdot u| - 2v\cdot u \geq 0$. For any $v,v_*,u,u_* \in \RR^3$, $\theta,\vartheta \in[0,\pi/2]$, write $\varphi = \varphi(v-v_*,u-u_*,\phi)$. Then:
\begin{align}
\label{eq:int-dphi}
\begin{split}
&\int_0^{2\pi} \left( |v + a(v,v_*,\theta,\phi) - u - a(u,u_*,\vartheta,\varphi)|^2 - |v-u|^2 \right) \frac{d\phi}{2\pi} \\
&= - \left[ \vphantom{\frac{1-\cos\theta}{2}} (v-u) + (v_*-u_*) \right] \cdot
     \left[\frac{1-\cos\theta}{2}(v-v_*) - \frac{1-\cos\vartheta}{2}(u-u_*) \right] \\
&\quad {} - \frac{\sin\theta \sin\vartheta}{4} R(v-v_*,u-u_*)
          + \frac{1-\cos(\theta-\vartheta)}{2} (v-v_*)\cdot(u-u_*).
\end{split}
\end{align}
\end{lem}

\begin{proof}
Setting
\begin{align*}
&b = \frac{v+v_*}{2} + \cos\theta \frac{v-v_*}{2},
\qquad r = \sin\theta \frac{|v-v_*|}{2},
\qquad d=\frac{v-v_*}{|v-v_*|}, \\
&\tilde{b} = \frac{u+u_*}{2} + \cos\vartheta \frac{u-u_*}{2},
\qquad \tilde{r} = \sin\vartheta \frac{|u-u_*|}{2},
\qquad \tilde{d}=\frac{u-u_*}{|u-u_*|},
\end{align*}
we have
\begin{align}
& 4(b-\tilde{b})^2 \notag \\
&= | (v-u) + (v_*-u_*) + \cos\theta (v-v_*) - \cos\vartheta (u-u_*)|^2 \notag \\
&= |v-u|^2 + |v_*-u_*|^2 + 2(v-u)\cdot (v_*-u_*) \notag \\
& \quad {} + \cos^2\theta|v-v_*|^2 + \cos^2\vartheta|u-u_*|^2 - 2\cos\theta\cos\vartheta(v-v_*)\cdot(u-u_*) \notag \\
& \quad {} + 2 [(v-u) + (v_*-u_*) ] \cdot [\cos\theta (v-v_*) - \cos\vartheta(u-u_*)] \notag \\
\begin{split}
\label{eq:4bb}
&= 3|v-u|^2 - |v_*-u_*|^2 + 2(v-u)\cdot (v_*-u_*) \\
& \quad {} + \cos^2\theta|v-v_*|^2 + \cos^2\vartheta|u-u_*|^2 - 2\cos\theta\cos\vartheta(v-v_*)\cdot(u-u_*) \\
& \quad {} - 2 [(v-u) + (v_*-u_*) ] \cdot [(1-\cos\theta)(v-v_*) - (1-\cos\vartheta)(u-u_*)]
\end{split}
\end{align}
and
\begin{align}
& 4(r-\tilde{r})^2 + 4r\tilde{r}(1-|d\cdot \tilde{d}|) \notag \\
&= \sin^2 \theta |v-v_*|^2 + \sin^2 \vartheta |u-u_*|^2 \notag \\
&\quad {} - \sin\theta \sin\vartheta (|v-v_*||u-u_*| + |(v-v_*)\cdot (u-u_*)|) \notag \\
\begin{split}
\label{eq:4rr}
&= \sin^2 \theta |v-v_*|^2 + \sin^2 \vartheta |u-u_*|^2 - \sin\theta\sin\vartheta R(v-v_*,u-u_*)\\
&\quad {} - 2 \sin\theta\sin\vartheta (v-v_*)\cdot(u-u_*).
\end{split}
\end{align}
Adding \eqref{eq:4bb} and \eqref{eq:4rr}, using that $|v-v_*|^2 + |u-u_*|^2 + 2(v-u)\cdot(v_*-u_*) = |v-u|^2 + |v_*-u_*|^2 + 2(v-v_*)\cdot(u-u_*)$ and the identity $\cos\theta\cos\vartheta + \sin\theta\sin\vartheta = \cos(\theta-\vartheta)$, yields
\begin{align*}
& 4(b-\tilde{b})^2 + 4(r-\tilde{r})^2 + 4r\tilde{r}(1-|d\cdot\tilde{d}|) \\
&= 4|v-u|^2  - \sin\theta\sin\vartheta R(v-v_*,u-u_*) \\
&\quad {} + 2(1-\cos(\theta-\vartheta)) (v-v_*)\cdot(u-u_*) \\
&\quad {} - 2[(v-u) + (v_*-u_*)] \cdot [(1-\cos\theta)(v-v_*) - (1-\cos\vartheta)(u-u_*)].
\end{align*}
Thanks to Lemma \ref{lem:optimal_circles}--\ref{lem:varphi}, $\varphi$ is an optimal transport map, and then the integral on the left side of \eqref{eq:int-dphi} without the term $-|v-u|^2$ is actually the cost given by Lemma \ref{lem:optimal_circles}--\ref{lem:W2circles}, that is, $(b-\tilde{b})^2 + (r-\tilde{r})^2 + r\tilde{r}(1-|d\cdot\tilde{d}|)$. Dividing by 4 and substracting $|v-u|^2$ in the above identity,  the result follows.
\end{proof}

\begin{cor} \label{cor:int-dphi-dz}
Assume \eqref{eq:nu}.
Fix any $K,L\in[0,\infty]$ with $K\geq L$, and define $\Phi_L^K :=  \int_L^K \frac{1-\cos G(z)}{2} dz  \geq 0$.
For any $v,v_*,u,u_*\in\RR^3$, write $\varphi = \varphi(v-v_*,u-u_*,\phi)$.
Then we have
\begin{align*}
& \int_0^\infty \int_0^{2\pi} \left( |v+c_K(v,v_*,z,\phi) - u - c_L(u,u_*,z,\varphi)|^2 - |v-u|^2 \right)
              \frac{d\phi}{2\pi} dz \\
&= \Phi_0^L[-|v-u|^2 + |v_*-u_*|^2] + \Phi_L^K(v-v_*)\cdot(2u-v-v_*) \\
& \qquad {} -  R(v-v_*,u-u_*) \int_0^L \frac{\sin^2 G(z)}{4}dz. \\
&\leq \Phi_0^L[-|v-u|^2 + |v_*-u_*|^2] + C(|v|+|v_*|+|u|)^2 (1+L)^{1-2/\nu}.
\end{align*}              
\end{cor}

\begin{proof}
Split the integral with respect to $z$ into $\int_0^L$ and $\int_L^K$. For the first integral we have $c_K(v,v_*,z,\phi) = a(v,v_*,\theta,\phi)$ and $c_L(u,u_*,z,\varphi) = a(u,u_*,\theta,\varphi)$ for $\theta = G(z)$; using Lemma \ref{lem:int-dphi}
yields the first and third terms in the equality. For the second integral we have $c_L(u,u_*,z,\varphi) = a(v,v_*,0,\varphi)$, so this time we use Lemma \ref{lem:int-dphi} with $\theta = G(z)$ and $\vartheta = 0$, which gives the second term. The inequality is then obtained discarding the negative third term, noting that
$\Phi_L^K \leq C \int_L^\infty G^2(z) dz$, and using \eqref{eq:boundG}.
\end{proof}

The next lemma is of key importance, since it gives a decoupling estimate
for the system of non-independent cutoff nonlinear processes $\mathbf{U}^L$. The proof,
also relying on a coupling argument, follows Lemma 6 of \cite{cortez-fontbona2016}.

\begin{lem}[decoupling] \label{lem:decoup}
Assume \eqref{eq:nu} and take $L\in[1,\infty)$. Then, there exists a constant $C$ independent of $L$ such
that for all $k\in\{1,\ldots,N\}$ and all $t\geq 0$,
\[
\W_2^2(\law^k(\mathbf{U}_t^L),(f_t^L)^{\otimes k} )
\leq C \frac{k}{N}.
\]
\end{lem}

\begin{proof}
Given $k\in\{1,\ldots,N\}$ fixed, we will construct $k$ independent cutoff nonlinear processes $\tilde{U}^{1,L},\ldots,\tilde{U}^{k,L}$ such that $\EE |U_t^{i,L} - \tilde{U}_t^{i,L}|^2$ is small, for all $i=1,\ldots,k$. To achieve this, the idea is the following: when $U_t^{i,L}$ has a simultaneous jump with some $U_t^{j,L}$ with $j\notin\{1,\ldots,k\}$, then the process $\tilde{U}_t^{i,L}$ will use the same  sample of $f_t^L$ used by $U_t^{i,L}$ to define its own jump; but when $j\in\{1,\ldots,k\}$, then one of the processes, $\tilde{U}_t^{i,L}$ or $\tilde{U}_t^{j,L}$, will not jump at that instant. We will then use an additional, independent source of randomness to define new jumps that compensate for the missing ones. Since, when $k \ll N$, the second kind of jump occurs much less frequently, this construction will give the desired estimate.

Consider a Poisson point measure $\M$ that is an independent copy of $\N$,
also independent from $(\mathbf{V}_0,\mathbf{U}_0)$, and define for each $i\in\{1,\ldots,k\}$
\begin{align*}
\M^i(dt,dz,d\phi,d\xi)
&= \N(dt,dz,d\phi,[i-1,i),d\xi) \\
&\quad {} + \N(dt,dz,d\phi,d\xi,[i-1,i)) \ind_{[k,N)}(\xi) \\
&\quad {} + \M(dt,dz,d\phi,d\xi,[i-1,i)) \ind_{[0,k)}(\xi).
\end{align*}
That is, $\M^i$ selects the atoms of $\N(dt,dz,d\phi,d\xi,d\zeta)$ where either $(\i(\xi) = i)$, or $(\i(\zeta) = i \text{ and } \i(\xi) \notin \{1,\ldots,k\})$, and to make up for the dropped atoms it also selects new ones from $\M(dt,dz,d\phi,d\xi,d\zeta)$, where $\i(\zeta) = i$ and $\i(\xi) \in \{1,\ldots,k\}$. This ensures that no such atom appears in two $\M^i$'s, implying that they are independent Poisson point measures, all with intensity $dt dz d\phi d\xi \ind_{A^i}(\xi)/[2(N-1)\pi]$, just like $\N^i$.

Mimicking \eqref{eq:Ui}, we define $\tilde{U}^{i,L}$ as the solution, starting from $\tilde{U}_0^{i,L} = U_0^i$, of the stochastic equation
\begin{equation} \label{eq:tildeUi}
d\tilde{U}_t^{i,L} = \int_0^\infty \int_0^{2\pi} \int_0^N c_L(\tilde{U}_{t^-}^{i,L},\Pi_t^{i,L}(\mathbf{U}_{t^-}^L,\xi), z,\tilde{\varphi}_{t^-}^i) \M^i(dt,dz,d\phi,d\xi).
\end{equation}
Here we write
\[
\tilde{\varphi}_{t^-}^i = \varphi(U_{t^-}^{i,L} - \Pi_t^{i,L}(\mathbf{U}_{t^-}^L,\xi), \tilde{U}_{t^-}^{i,L} - \Pi_t^{i,L}(\mathbf{U}_{t^-}^L,\xi), \varphi_{t^-}^i),
\]
where $\varphi_{t^-}^i$ was defined in \eqref{eq:varphi_t}. In other words: $\tilde{\varphi}_{t^-}^i$ takes the angle $\varphi_{t^-}^i$ and maps it to $[0,2\pi)$ in such a way that the resulting pair $(\varphi_{t^-}^i,\tilde{\varphi}_{t^-}^i)$ parametrizes (as a function of $\phi$) an optimal coupling of the uniform distributions on the circles with orthogonal directions $U_{t^-}^{i,L} - \Pi_t^i(\mathbf{U}_{t^-}^L,\xi)$ and $\tilde{U}_{t^-}^{i,L} - \Pi_t^i(\mathbf{U}_{t^-}^L,\xi)$, whenever $\phi$ is uniformly chosen on $[0,2\pi)$. The latter  ensures closeness of the states of $U_t^{i,L}$ and $\tilde{U}_t^{i,L}$ after the joint jump.

If we define $\tilde{\M}^{i,L}(dt,dz,d\phi,dv)$
to be the point measure on $[0,\infty)\times[0,\infty)\times [0,2\pi)\times \RR^3 $   
with atoms $(t,z,\tilde{\varphi}_{t^-}^i,\Pi_t^{i,L}(\mathbf{U}_{t^-}^L,\xi))$
for every atom $(t,z,\phi,\xi)$ of $\M^i$,
it is clear that $\tilde{U}^{i,L}$ depends only on $\tilde{\M}^{i,L}$ and $U_0^i$.
Since: 1) the dependence on $\mathbf{V}$ and $\mathbf{U}^L$ in \eqref{eq:tildeUi} is predictable with respect to $\N$, $\M$ and the initial data,
2) the Poisson measures $\M^1,\ldots,\M^k$ are independent,
3) the $\xi$-law of $\Pi_t^{i,L}(\mathbf{x},\xi)$ is $f_t^L$ for every $\mathbf{x}\in\RR^N$, and
4) the $\phi$-law of $\varphi(v,u,\phi)$ is the uniform distribution
on $[0,2\pi)$ for any $v,u\in\RR^3$,
one can use the compensation formula to compute the joint
Laplace functional of $\tilde{\M}^{1,L},\ldots,\tilde{\M}^{k,L}$
and deduce that they are independent Poisson point measures,
all with intensity $dt dz d\phi f_t^L(dv)/2\pi$.
This shows that $\tilde{U}^{i,L}$ satisfies \eqref{eq:nonlinear_cutoff}
with $\M^L$ replaced by $\tilde{\M}^{i,L}$,
and then $\tilde{U}^{1,L},\ldots,\tilde{U}^{k,L}$ are independent cutoff nonlinear processes.

Consequently, we have
\[
\W_2^2(\law^k(\mathbf{U}_t^L), (f_t^L)^{\otimes k} )
\leq \EE \frac{1}{k} \sum_{j=1}^k |U_t^{j,L} - \tilde{U}_{t}^{j,L}|^2
= \EE|U_t^{i,L} - \tilde{U}_t^{i,L}|^2.
\]
Thus, it suffices to estimate the quantity $h_t := \EE|U_t^{i,L} - \tilde{U}_t^{i,L}|^2$, for any $i \in\{1,\ldots,k\}$ fixed. We can write
\begin{equation} \label{eq:J1+J2+J3}
h_t - h_s
= J_{s,t}^1 + J_{s,t}^2 + J_{s,t}^3
\end{equation}
where $J_{s,t}^1$ is the term associated with simultaneous jumps of $U^{i,L}$ and $\tilde{U}^{i,L}$, $J_{s,t}^2$ corresponds to jumps of $U^{i,L}$ alone, and $J_{s,t}^3$ corresponds to jumps of $\tilde{U}^{i,L}$ alone. To write this terms explicitly, let us first shorten notation: write $U_r := U_r^{i,L}$, $\tilde{U}_r := \tilde{U}_r^{i,L}$, $\Pi_r := \Pi_r^{i,L}(\mathbf{U}_r^L,\xi)$, $c_r := c(U_r,\Pi_r,z,\varphi_r^i)$ and $\tilde{c}_r := c(\tilde{U}_r,\Pi_r,z,\tilde{\varphi}_r^{i})$. From \eqref{eq:Ui} and \eqref{eq:tildeUi}, $J_{s,t}^1$, $J_{s,t}^2$ and $J_{s,t}^3$ are thus given by
\begin{align*}
J_{s,t}^1
&= \EE \int_{(s,t]} \int_0^L \int_0^{2\pi} \int_0^N
   \left( | U_{r^-} + c_{r^-} - \tilde{U}_{r^-} - \tilde{c}_{r^-} |^2
          - | U_{r^-} - \tilde{U}_{r^-} |^2 \right) \\
&\qquad {} \times \left[ \vphantom{|U_{r^-}^i|^2} \N(dr,dz,d\phi,[i-1,i),d\xi)
                       + \N(dr,dz,d\phi,d\xi,[i-1,i)) \ind_{[k,N)}(\xi) \right], \\
J_{s,t}^2
&= \EE \int_{(s,t]} \int_0^L \int_0^{2\pi} \int_0^N
   \left( | U_{r^-} + c_{r^-} - \tilde{U}_{r^-} |^2
          - | U_{r^-} - \tilde{U}_{r^-} |^2 \right) \\
&\qquad {} \times \N(dr,dz,d\phi,d\xi,[i-1,i)) \ind_{[0,k)}(\xi), \\
J_{s,t}^3
&= \EE \int_{(s,t]} \int_0^L \int_0^{2\pi} \int_0^N
   \left( | U_{r^-} - \tilde{U}_{r^-} -\tilde{c}_{r^-} |^2
          - | U_{r^-} - \tilde{U}_{r^-} |^2 \right) \\
&\qquad {} \times \M(dr,dz,d\phi,d\xi,[i-1,i)) \ind_{[0,k)}(\xi).
\end{align*}
Recall that $\N$ and $\M$ have intensity  $dt dz d\phi d\xi d\zeta \ind_{\G}(\xi,\zeta)/[4(N-1)\pi]$. Note that $\int_{i-1}^i \ind_{\G}(\xi,\zeta) d\zeta = \ind_{A^i}(\xi)$, where $A^i = [0,N)\setminus [i-1,i)$. Using the compensation formula, the Poisson point measures in the integrals can be replaced by their  intensities, and we obtain for $J_{s,t}^1$:
\begin{align}
J_{s,t}^1
\notag &= \EE \int_s^t \int_0^L \int_0^{2\pi} \int_0^N \left( | U_r + c_r - \tilde{U}_r - \tilde{c}_r |^2
          - | U_r - \tilde{U}_r |^2 \right)  [ \ind_{A^i}(\xi) + \ind_{[k,N)}(\xi) ]
           \frac{dr dz d\phi d\xi}{4(N-1)\pi}  \\
\notag
&\leq - \Phi_0^L \EE \int_s^t \int_0^N |U_r - \tilde{U}_r|^2
         [ \ind_{A^i}(\xi) + \ind_{[k,N)}(\xi) ] \frac{dr d\xi}{2(N-1)} \\
\label{eq:J1}
&\leq -\frac{\Phi_0^L}{2} \int_s^t h_r dr,
\end{align}
where we have used Corollary \ref{cor:int-dphi-dz} with $v=U_r$, $u=\tilde{U}_r$, $v_* = u_* = \Pi_r$, $K=L$ and the change of variable $(\phi,\varphi_r^i)\mapsto (\varphi_r^i,\tilde{\varphi}_r^i)$. For $J_{s,t}^2$ we get:
\begin{align}
\notag
J_{s,t}^2
&= \EE \int_s^t \int_0^L \int_0^{2\pi} \int_0^N \left( | U_r + c_r - \tilde{U}_r|^2
          - | U_r - \tilde{U}_r |^2 \right) \ind_{A^i \cap [0,k)}(\xi) \frac{dr dz d\phi d\xi}{4(N-1)\pi}  \\
\notag
&\leq C \EE \int_s^t \int_0^N (|U_r| + |\Pi_r| + |\tilde{U}_r|)^2
           \ind_{A^i \cap [0,k)}(\xi) \frac{dr d\xi}{2(N-1)} \\
&\leq C (t-s)\frac{k-1}{N-1} \label{eq:J2},
\end{align}
where in the second step we have used Corollary \ref{cor:int-dphi-dz} again, with $v=U_r$, $u=\tilde{U}_r$, $v_* = u_* = \Pi_r$ and the roles of $K$ and $L$ exchanged (the smallest one being equal to 0).
In the last step we have used the fact that $f_r^L$ has uniformly bounded moments of order $2$,
that $\law(U_r)=\law(\tilde{U}_r) = f_r^L$,
and that $\EE \int_{j-1}^j |\Pi_r|^2 d\xi = \int_{\RR^3} |u|^2 f_r^L(du)$ for all $j\neq i$,
thanks to Lemma \ref{lem:coup}.
Similarly for $J_{s,t}^3$: using Corollary \ref{cor:int-dphi-dz} with $v = \tilde{U}_r$, $u=U_r$,
$v_* = u_* = \Pi_r$, the bound \eqref{eq:J2} is also valid for $J_{s,t}^3$.

Thus, from \eqref{eq:J1+J2+J3}, \eqref{eq:J1} and \eqref{eq:J2}, we obtain
that $\partial_t h_t \leq -(\Phi_0^L/2) h_t + Ck/N$ for almost all $t\geq 0$,
and, since $h_0 = 0$, the conclusion follows from Gronwall's lemma
(the dependence on $L$ can be dropped since $\Phi_0^L$ is bounded away
from $0$ when $L\geq 1$ thanks to the lower bound in \eqref{eq:boundG}).
\end{proof}

%%%%%%%%%%%%%%%%%%%%%%%%%%%%%%%%%%%%%%%%%%%%%%%%%%%%%%%%
%%%%%%%%%%%%%%%%%%%%%%%%%%%%%%%%%%%%%%%%%%%%%%%%%%%%%%%%
\section{Proof of Theorem \ref{thm:main}} \label{sec:proof_main}
%%%%%%%%%%%%%%%%%%%%%%%%%%%%%%%%%%%%%%%%%%%%%%%%%%%%%%%%
%%%%%%%%%%%%%%%%%%%%%%%%%%%%%%%%%%%%%%%%%%%%%%%%%%%%%%%%

For a probability measure $\mu$ on $\RR^d$, call $\varepsilon_n (\mu) := \EE \W_2^2(\mu, \bar{\mathbf{Z}})$,
where $\mathbf{Z} = (Z^1,\ldots,Z^n) \in (\RR^d)^n$ is a vector of $n$ independent and $\mu$-distributed
random variables on $\RR^d$. The best general estimate available for $\varepsilon_n(\mu)$
is the following, whose proof can be found in \cite{fournier-guillin2013}:
for any $p>4$, there exists a constant $C_p<\infty$ such that for every $\mu \in \P(\RR^d)$
\begin{equation} \label{eq:epsnmu}
\varepsilon_n(\mu) \leq \frac{C_p (\int |v|^p \mu(dv) )^{2/p}}{n^{1/2}}.
\end{equation}

The following lemma will allow us to work with $\W_2^2(\law^n(\mathbf{U}_t^L), (f_t^L)^{\otimes n})$ instead of $\EE\W_2^2(\bar{\mathbf{U}}_t^{i,L},f_t^L)$, but at the price of the extra term $\varepsilon_n(f_t^L)$:

\begin{lem} \label{lem:EW2overY}
Let $\mathbf{X} = (X^1,\ldots,X^m) \in (\RR^d)^m$ be an exchangeable random vector, and let $\mu\in\P(\RR^d)$. Then, for any $n\leq m$,
\begin{align*}
\frac{1}{2} \EE \W_2^2(\bar{\mathbf{X}}, \mu)
&\leq \frac{kn}{m} \left( \W_2^2(\law^n(\mathbf{X}),\mu^{\otimes n}) + \varepsilon_n(\mu) \right) \\
&\quad {} + \frac{\ell}{m} \left( \W_2^2(\law^\ell(\mathbf{X}),\mu^{\otimes \ell}) + \varepsilon_\ell(\mu) \right),
\end{align*}
where $k$ and $\ell$ are the unique non-negative integers satisfying $m = kn + \ell$ and $\ell\leq n-1$.
\end{lem}

\begin{proof}
See the proof of Lemma 7 of \cite{cortez-fontbona2016}.
\end{proof}

\begin{cor} \label{cor:EEW2overU}
Assume \eqref{eq:nu} and \eqref{eq:p0}, and take $L\in[1,\infty)$. Then, there exists a constant $C>0$
independent of $L$, such that
for any $i\in\{1,\ldots,N\}$ and for all $t\geq 0$,
\[
\EE \W_2^2(\bar{\mathbf{U}}_t^{i,L},f_t^L) \leq C N^{-1/3}.
\]
Moreover, the same bound is valid with $\bar{\mathbf{U}}_t^L$ in place of $\bar{\mathbf{U}}_t^{i,L}$.
\end{cor}

\begin{proof}
In the notation of Lemma \ref{lem:EW2overY}, set $m=N-1$, $\mathbf{X} = (U_t^{j,L})_{j\neq i}$ and $\mu = f_t^L$ and given  $n\leq m$, simply bound $\frac{kn}{m} \leq 1$, $\frac{\ell}{m} \leq \frac{n}{N}$,
$\W_2^2(\law^\ell(\mathbf{X}),\mu^{\otimes \ell}) \leq 4 E $ and $\varepsilon_\ell(\mu) \leq 4 E$, where $E= \int |v|^2 f_t^L(dv) = \int |v|^2 f_0(dv)$. Using that result we get
\begin{align*}
\frac{1}{2} \EE \W_2^2(\bar{\mathbf{U}}_t^{i,L},f_t^L)
&\leq \W_2^2(\law^n(\mathbf{U}_t^L), (f_t^L)^{\otimes n}) + \varepsilon_n(f_t^L)
     + 8E\frac{n}{N} \\
&\leq C \frac{n}{N} + C \frac{1}{n^{1/2}},
\end{align*}
where we have used Lemma \ref{lem:decoup} together with \eqref{eq:epsnmu}
with $p=p_0>4$ and the
uniformity of the $p_0$-moments of $f_t^L$.
Choosing $n = \lfloor N^{2/3} \rfloor$ yields
the desired result.
To obtain the same estimate with $\bar{\mathbf{U}}_t^L$ on the left hand side, use $m=N$ and $\mathbf{X}=\mathbf{U}_t^L$.
\end{proof}

We need to make sure that the cutoff $L$ can be removed
in a satisfactory manner; for instance, we can use Theorem 5.2
of \cite{fournier-meleard2002}. However, for the reader's convenience,
we state here a result specific for the cutoff we use, with a shorter proof
and better dependence on time:

\begin{lem}
\label{lem:ftL}
Assume \eqref{eq:nu} and \eqref{eq:p0}. Then there exists a constant $C$ such that for all $t\geq 0$ and
all $L>0$,
\[
\W_2^2(f_t^L,f_t) \leq CtL^{1-2/\nu} .
\]
\end{lem}

\begin{proof}
If $f_0$ is a Dirac mass, then $f_t = f_t^L = f_0$ for all $t\geq 0$
and the result is trivial.
If $f_0$ is not a Dirac mass, 
we know that $f_t$ has a density for $t>0$ thanks to Theorem \ref{thm:well_posed};
therefore, there exists an optimal transport map $T_t^L:\RR^3\to\RR^3$
such that for any random vector $X \in \RR^3$ with law $f_t$, the pair
$(X,T_t^L(X))$ is an optimal coupling between $f_t$ and $f_t^L$.
Moreover, thanks to the measurability of the flows $t\mapsto (f_t,f_t^L)$ and to Theorem 1.1 in  \cite{fontbona-guerin-meleard2010}, 
 the maps $T_t^L$ can be chosen in such a a way that the mapping
$(t,v) \mapsto (v,T_t^L(v))$ is measurable.
Now, given a (weak) solution $(W,\M)$ to \eqref{eq:nonlinear},
define a process $W^L$ as the unique, jump-by-jump solution, starting from $W_0^L = W_0$, to
the stochastic equation
\[
dW_t^L = \int_0^\infty \int_0^{2\pi} \int_{\RR^3} c_L(W_{t^-}^L,T_t^L(v),z,\varphi_{t^-}) \M(dt,dz,d\phi,dv),
\]
with $\varphi_{t^-} = \varphi(W_{t^-} - v, W_{t^-}^L - T_t^L(v),\phi)$.
Arguing as in the proof of part \ref{lemaSDE_iii}  of Lemma \ref{lemaSDE},
one can verify that $W^L$ is a cutoff nonlinear process,
and in particular, 
 $W_t^L$ has law $f_t^L$ for each $t\geq 0$. 
For $h_t = \EE |W_t-W_t^L|^2$, we obtain from this and \eqref{eq:nonlinear}:
\begin{align*}
\partial_t h_t
&= \EE \int_0^\infty \int_0^{2\pi} \int_{\RR^3}
\left( |W_t + c(W_t,v,z,\phi) - W_t^L - c_L(W_t^L,T_t^L(v),z,\varphi_t)|^2 \right. \\
& \qquad\qquad\qquad\qquad \left. {} - |W_t - W_t^L|^2 \right) \frac{dtd\phi f_t(dv)}{2\pi} \\
& \leq \EE \int_{\RR^3} \left( -\Phi_0^L |W_t - W_t^L|^2 + \Phi_0^L|v - T_t^L(v)|^2 \right. \\
& \qquad\qquad\qquad\qquad \left. {} + C(|W_t|+|v|+|W_t^L|)^2 L^{1-2/\nu} \right) f_t(dv),
\end{align*}
where we have used Corollary \ref{cor:int-dphi-dz} with $K=\infty$.  By construction, we have $\int_{\RR^3} |v-T_t^L(v)|^2 f_t(dv)
= \W_2^2(f_t^L,f_t) \leq h_t$, and using the preservation of the second moment
for $f_t^L$ and $f_t$, the last inequality yields $\partial_t h_t \leq C L^{1-2/\nu}$.
Since $h_0=0$, the result follows.
\end{proof}

We are now ready to prove Theorem \ref{thm:main}.

\begin{proof}[Proof of Theorem \ref{thm:main}]
Take $L \in [1,\infty)$.
For some $i\in\{1,\ldots,N\}$ fixed, we will estimate the quantity
$h_t := \EE|V_t^i-U_t^{i,L}|^2$.
To shorten notation, call $V_r := V_r^i$, $V_r^\i := V_r^{\i(\xi)}$,
$U_r := U_r^{i,L}$ and $\Pi_r := \Pi^{i,L}_r(\mathbf{U}_r^L,\xi)$.
From \eqref{eq:Vi} and \eqref{eq:Ui} it follows that for every $0\leq s <t$,
\begin{align*}
h_t - h_s
&= \EE \int_{(s,t]} \int_0^\infty \int_0^{2\pi} \int_0^N
   \left( \vphantom{\left|V_{r^-} \right|^2}
   \left|V_{r^-} + c(V_{r^-}, V_{r^-}^{\i},z,\phi) - U_{r^-} \right. \right. \\
& \quad \left. \left. \vphantom{V_{r^-}^{\i}}
   {}  - c_L(U_{r^-}, \Pi_{r^-}, z, \varphi_{r^-}^i)  \right|^2 
   - \left|V_{r^-} - U_{r^-} \right|^2 \right) \N^i(dr,dz,d\phi,d\xi).
\end{align*}
Using the compensation formula, $\N^i$ can be replaced by its intensity $dr dz d\phi d\xi \ind_{A^i}(\xi) / [2(N-1)\pi]$, where $A^i = [0,N)\setminus [i-1,i)$.  Corollary \ref{cor:int-dphi-dz} with $v=V_r$, $v_*=V_r^\i$,
$u=U_r$, $u_* = \Pi_r $ and $K=\infty$ yields 
\begin{align*}
h_t - h_s
&\leq \EE \int_s^t  \int_{A^i} \left[-\Phi_0^L |V_r - U_r|^2 + \Phi_0^L |V_r^\i - \Pi_r|^2 \right. \\
& \qquad \qquad \left. 
           + C(|V_r| + |V_r^{\i}| + |U_r|)^2 L^{1-2/\nu}\right] \frac{dr d\xi}{N-1}.
\end{align*}
Note that
$|V_r^\i - \Pi_r|^2 \leq |V_r^\i - U_r^\i|^2 + 2|V_r^\i - U_r^\i||U_r^\i - \Pi_r| + |U_r^\i - \Pi_r^i|^2$,
where $U_r^\i := U_r^{\i(\xi),L}$, and that $\EE \int_{A^i} |V_r^\i - U_r^\i|^2 d\xi /(N-1) = h_r$
by exchangeability. Also, thanks to Lemma \ref{lem:coup},
we know that $ \int_{A^i} |U_r^\i - \Pi_r|^2 d\xi/(N-1) = \W_2^2(\bar{\mathbf{U}}_r^{i,L},f_r^L)$.
Calling $g_t := \EE\W_2^2(\bar{\mathbf{U}}_t^{i,L},f_t^L)$, using the Cauchy-Schwarz inequality and the fact
that both the particles and the cutoff nonlinear processes have uniformly bounded second moment, we obtain
$\partial_t h_t \leq C[h_t^{1/2} g_t^{1/2} + g_t + L^{1-2/\nu}]$
for almost every $t\geq 0$.
Using a version of Gronwall's lemma (see for instance Lemma 4.1.8 of \cite{ambrosio-gigli-savare2008})
we deduce that $h_t \leq C [ h_0 + t(1+t)g_t + t L^{1-2/\nu}]$,
and consequently,
\[
h_t \leq C [h_0 + t(1+t)N^{-1/3} + t L^{1-2/\nu}],
\]
where we have used
Corollary \ref{cor:EEW2overU} to bound $g_t \leq CN^{-1/3}$ uniformly on $t$ and $L$. From this, we obtain for all $L\geq 1$
\begin{align*}
\EE \W_2^2(\bar{\mathbf{V}}_t,f_t)
&\leq C [\EE \W_2^2(\bar{\mathbf{V}}_t,\bar{\mathbf{U}}_t^L)
  + \EE\W_2^2(\bar{\mathbf{U}}_t^L,f_t^L) 
  + \W_2^2(f_t^L,f_t) ] \\
&\leq C[h_t +  N^{-1/3} +  tL^{1-2/\nu}]\\
&\leq C[h_0 + (1+t)^2 N^{-1/3} + tL^{1-2/\nu}],
\end{align*}
where we have used Corollary \ref{cor:EEW2overU} again, together with Lemma \ref{lem:ftL}. Letting $L\to\infty$, the result follows.
\end{proof}

%%%%%%%%%%%%%%%%%%%%%%%%%%%%%%%%%%%%%%%%%%%%%%%%%%%%%%%%
%%%%%%%%%%%%%%%%%%%%%%%%%%%%%%%%%%%%%%%%%%%%%%%%%%%%%%%%
\section{Uniform propagation of chaos} \label{sec:unif_chaos}
%%%%%%%%%%%%%%%%%%%%%%%%%%%%%%%%%%%%%%%%%%%%%%%%%%%%%%%%
%%%%%%%%%%%%%%%%%%%%%%%%%%%%%%%%%%%%%%%%%%%%%%%%%%%%%%%%

In this section we give the proof of Theorem \ref{thm:uniform_Maxwell}.

The following is a version of Povzner's lemma \cite{povzner1962}, 
see for instance \cite{elmroth1983,wennberg1997,gamba-panferov-villani2004,mischler-wennberg1999} and the references therein for other versions.
Plainly, it will be crucial to establish the
propagation of moments for the  particle system, needed to take full advantage of the stability result in \cite{rousset2014}. A proof can be found in  \cite{mischler-wennberg1999} in a slightly different setting; for the readers convenience,  we provide a proof  of the precise statement below  in the Appendix in the case of  $p$ even (which is enough for purposes).

%%%%%%% Povzner
\begin{lem}[a version of Povzner's lemma]
\label{lem:Povzner}
Assume that $I := \int_0^{\pi/2} \theta^2 \beta(\theta)d\theta < \infty$.
Then, for any $p> 2$ and any $v,v_* \in\RR^3$, we have
\begin{align*}
&\int_{\SS^2} ( |v'|^p + |v_*'|^p - |v|^p - |v_*|^p) B(\theta) d\sigma \\
&\leq -A_p (|v|^p + |v_*|^p) + I \tilde{A}_p ( |v|^{p-2} |v_*|^2 + |v_*|^{p-2}|v|^2),
\end{align*}
where $\tilde{A}_p>0$ is some constant that depends only on $p$, and
\[
A_p := \int_0^{\pi/2} [ 1 - \sin(\theta/2)^p  - \cos(\theta/2)^p ] \beta(\theta)d\theta >0.
\]
\end{lem}

The propagation of moments for the particle system was already established in \cite[Lemma 5.3]{mischler-mouhot2013}, where it is assumed that the initial energy is a.s.\ bounded. In the Maxwellian setting, following the proof of that lemma and performing a careful inspection of how the constants depend on the energy, we are able to write a moments estimate where this dependence is explicit. More importantly,  using conditional expectations we  are able to deduce propagation of moments for a general initial condition, without any  restriction (besides  of course finite second moment). This is stated in the next result:

\begin{cor}[propagation of moments for the particle system] \label{cor:unif_moments}
Assume \eqref{eq:nu},  let $p\geq 2$ and define the random variable $\mathcal{E}:=\frac{1}{N}\sum_i^N |V_0^i|^2 $.
 Then, there exists a constant $C_p>0$  (nonrandom and  not depending on $\mathbf{V}_0$ nor on its law)  such that for all $t\geq 0$, 
\[
\EE \left( |V_t^1|^p\,  \middle|  \, \mathcal{E} \right) \leq \EE\left(  |V_0^1|^p\,  \middle|  \, \mathcal{E} \right)   + C_p \mathcal{E}^{p/2}, \qquad a.s.
\]
As a consequence, we have
\[
\EE |V_t^1|^p \leq (C_p+1) \, \EE |V_0^1|^p, \qquad \forall t\geq 0.
\]
\end{cor}

\begin{proof} For each $t\geq 0$, write $g_t = \EE \left( |V_t^1|^p\,  \middle|  \, \mathcal{E} \right) $ which is a.s.\ equal to   $ \EE \left(\frac{1}{N}\sum_i^N    |V_t^i|^p\,  \big|  \, \mathcal{E} \right) $ by exchangeability of the system $\mathbf{V}$. 
Since the process 
\begin{align*}
&\frac{1}{N}\sum_i^N  |V^i_t|^p - \frac{1}{N}\sum_i^N  |V^i_0|^p \\
&\qquad\qquad -  \frac{1}{N}\sum_i^N  \int_0^t \int_0^\infty \int_0^{2\pi} \int_{A^i} 
  [|V_s^i + c(V_s^i,V_s^{\i(\xi)},z,\phi) |^p - |V_s^i|^p]  \frac{d\xi d\phi dz ds}{2(N-1)\pi}
  \end{align*}
is a martingale in the filtration $({\cal F}_t)_{t\geq 0}$ defined in Section  \ref{subsect:The particle system} and  $\mathcal{E}$ is ${\cal F}_0$ measurable, taking into account  equation \eqref{eq:int_B} we get  for each $t\geq 0$ 
\begin{align*}
 g_t - g_0
&= \EE \left( \frac{1}{N}\sum_i^N  \int_0^t \int_0^\infty \int_0^{2\pi} \int_{A^i}
  [|V_t^i + c(V_t^i,V_t^{\i(\xi)},z,\phi) |^p - |V_t^i|^p]  \frac{d\xi d\phi dz dt}{2(N-1)\pi} \, \middle|  \, \mathcal{E} \right) \\
&= 
  \int_0^t \EE \left(  \frac{1}{N^2}\sum_i^N  \sum_{j\neq i}    \int_0^{2\pi} \int_0^{\pi/2} [|V'^i_t|^p + |V'^j_t|^p - |V_t^i|^p  - |V_t^j|^p]
  \beta(\theta) \frac{d\theta d\phi}{4\pi} \, \middle|  \, \mathcal{E} \right)  dt
  \end{align*}
almost surely.   The latter implies that $t\mapsto g_t$ has an a.s.\ absolutely continuous version which we work with from now on. Taking the difference
$g_t - g_s$, dividing by $t-s$ and letting $s\to t$,
by  Lemma \ref{lem:Povzner} we a.s.\ have
for some positive constants $I$, $A_p$ and $\tilde{A}_p$ and almost every $t\geq 0$ :
\begin{align*}
\partial_t g_t
& \leq - A_p g_t + I \tilde{A}_p  \EE \left(   \frac{1}{N}\sum_i^N  |V_t^i|^{p-2} \frac{1}{N-1}\sum_{j\neq i} |V_t^j|^2  \, \middle|  \, \mathcal{E} \,  \right)  \\
& \leq - A_p g_t +  2 \mathcal{E} I \tilde{A}_p g_t^{1-2/p},
\end{align*}
where we have used that $\sum_{k\neq j} |V_t^k|^2 \leq \sum_{k} |V_t^k|^2 = N \mathcal{E}$ since the system almost surely preserves energy, together with exchangeability and  the conditional H\"older inequality. This differential inequality
 implies that $g_t \leq \max(g_0, x^*)$ for all $t\geq 0$,
where $x^* = (2 \mathcal{E} I \tilde{A}_p / A_p)^{p/2} $ is the unique positive root 
of the polynomial $- A_p x +  2 \mathcal{E} I \tilde{A}_p x^{1-2/p}$. This implies
\[
g_t
\leq g_0 + C_p \mathcal{E}^{p/2},
\]
for some constant depending only on $A_p$, $\tilde{A}_p$  and $I$, which proves the first statement.  For the second one,  we use conditional expectation to get
\begin{align*}
\EE |V_t^1|^p
& \leq  \EE \left[ \EE\left(  |V_0^1|^p\, \middle|  \, \mathcal{E} \right)   + C_p \mathcal{E}^{p/2}  \right] \\
&= \EE |V_0^1|^p + C_p \EE \left[ \left( \frac{1}{N} \sum_{k=1}^N |V_0^k|^2 \right)^{p/2} \right],
\end{align*}
and then Jensen's inequality (applied in the empirical mean) gives $\EE |V_t^1|^p \leq \EE |V_0^1|^p + C_p \EE \frac{1}{N} \sum_{k=1}^N |V_0^k|^p$. The proof is complete.
\end{proof}

The fact that the particles have bounded moments allows us to obtain
a convergence to equilibrium result for the particle
system that extends Proposition 1.7 of \cite{rousset2014}. We will state it in terms of the following distance:
define $\mathscr{W}_2$ as the usual $2$-Wasserstein distance on $\P(\P(\RR^3))$ induced by $\W_2$, that is
\[
\mathscr{W}_2(\alpha,\beta)
= \inf_{\mathbf{a},\mathbf{b}} \left( \EE \W_2^2(\mathbf{a},\mathbf{b}) \right)^{1/2},
\]
where the infimum is taken over all random elements $\mathbf{a}$ and $\mathbf{b}$ in $\P(\RR^3)$ such that $\law(\mathbf{a}) = \alpha$ and $\law(\mathbf{b}) = \beta$. For $F\in \P_2^{\textnormal{sym}}((\RR^3)^N)$, denote by  $\hat{F}$ the push-forward of $F$ by the ``empirical measure'' map, that is, $\hat{F} = \law(\bar{\mathbf{X}}) \in \P(\P(\RR^3))$ for $\mathbf{X} \sim F$.
It is also clear that for any $\mathbf{X} \sim F$ and $\mu\in\P(\RR^3)$, we have
\begin{equation}
\label{eq:EW2W2hatF}
\EE \W_2^2(\bar{\mathbf{X}},\mu)
= \mathscr{W}_2^2(\hat{F},\delta_\mu).
\end{equation}

Define also the \emph{Boltzmann sphere}
\[
\mathcal{S}^N = \left\{\mathbf{x} \in (\RR^3)^N : \frac{1}{N}\sum_{i=1}^N x^i = 0, \frac{1}{N}\sum_{i=1}^N |x^i|^2 = 1 \right\},
\]
and notice that if $G_0^N$ is concentrated on $\mathcal{S}^N$, then the preservation
of momentum and kinetic energy of the collisions imply that $G_t^N := \law(\mathbf{V}_t)$ is also concentrated
on $\mathcal{S}^N$ for all $t\geq 0$. Denote $\mathcal{U}^N$ the uniform distribution
on $\mathcal{S}^N$.
Theorem 1.6 of \cite{rousset2014} states that for any $\delta>0$ and $q>1$,
\begin{equation}
\label{eq:rousset_thm}
\partial_t^+ \mathscr{W}_2(\hat{G}_t^N,\hat{\mathcal{U}}^N)
\leq -c_{\delta,q,N}(G_t^N) \mathscr{W}_2(\hat{G}_t^N,\hat{\mathcal{U}}^N)^{1+1/\delta},
\end{equation}
where $c_{\delta,q,N}(G_t^N) = k_{\delta,q} \EE(|V_t^1|^{2q(1+\delta)})^{-1/2q\delta}$
and $k_{\delta,q}>0$ is some constant. We are now ready to state and prove:

\begin{lem}
\label{lem:rousset}
Assume \eqref{eq:nu} and that $G_0^N \in \P_2^{\textnormal{sym}}((\RR^3)^N)$ is concentrated on the
Boltzmann sphere $\mathcal{S}^N$.
Assume also that $R_p := \sup_N \EE |V_0^1|^p < \infty$ for some $p\geq 4$. Then,
for all $0<\delta <p-2$ there exists a constant $C_{p,\delta}>0$ depending only
on $p$, $\delta$ and $R_p$, such that for all $N\in\NN$
and all $t\geq 0$
\[
\mathscr{W}_2^2(\hat{G}_t^N,\hat{\mathcal{U}}^N)
\leq \left[\mathscr{W}_2^2(\hat{G}_0^N,\hat{\mathcal{U}}^N)^{-1/\delta} + C_{p,\delta} t \right]^{-\delta}.
\]
\end{lem}

\begin{proof}
Take $q>1$ and $\delta<p/2-1$ such that
$2q(1+\delta) = p$. Using Corollary \ref{cor:unif_moments}, we deduce that
$c_{\delta,q,N}(G_t^N) \geq k_{\delta,q} ( C_p \EE|V_0^1|^p )^{-1/2q\delta}
\geq k_{\delta,q} ( C_p R_p )^{-1/2q\delta} =: C_{p,\delta}$.
From \eqref{eq:rousset_thm}, the result follows using Gronwall's lemma,
squaring, and redefining $\delta$ as $2\delta$.
\end{proof}

We conclude with the proof of Theorem \ref{thm:uniform_Maxwell}.
We follow a  standarization argument found in \cite{fournier-guillin2015}, which allows one to reduce
the proof to the case where the initial distribution $G_0^N$ is concentrated on the Boltzmann sphere $\mathcal{S}^N$.

Given any $F \in \P^\textnormal{sym}((\RR^3)^N)$ and a random vector $\mathbf{X} \sim F$, set
$\st{F} = \law(\st{\mathbf{X}})$, where $\st{\mathbf{X}} = \mathbf{Y} = (Y^1,\ldots,Y^N)$ is defined as
\begin{equation}
\label{eq:Yi}
%\begin{split}
Y^i = \frac{X^i - M}{S}, \quad \text{ with }
\quad M = \frac{1}{N}\sum_j X^j,
\quad S^2 = \frac{1}{N} \sum_j |X^j-M|^2
%\end{split}
\end{equation}
on the event $\{S>0\}$, and $Y^i = Z^i$ on $\{S = 0\}$, where $\mathbf{Z}$
is some (arbitrary but fixed) exchangeable random vector on $\mathcal{S}^N$, independent of $\mathbf{X}$
(although we will mainly use this standarization when $S>0$).
This, of course, ensures that $\st{F}$ is symmetric and concentrated on $\mathcal{S}^N$.
Calling $Q^2 = \frac{1}{N}\sum_i |X^i|^2$ we have $S^2 = Q^2 - |M|^2$,
and then
\begin{align}
\frac{1}{N}\sum_i |X^i-Y^i|^2
&= \ind_{\{S>0\}} \frac{1}{N} \sum_i \left| \frac{(S-1)X^i + M}{S} \right|^2 
  + \ind_{\{S=0\}} \frac{1}{N} \sum_i |M - Z^i|^2 \notag \\
&= \ind_{\{S>0\}} \frac{(S-1)^2 Q^2  + (2S-1)|M|^2}{S^2}
  + \ind_{\{S=0\}} [1 + |M|^2] \notag \\
&= \ind_{\{S>0\}} [(S-1)^2 + |M|^2]
   + \ind_{\{S=0\}} [1 + |M|^2] \notag \\
&= (S-1)^2 + |M|^2
\leq \W_2^2(\bar{\mathbf{X}}, \mu), \label{eq:sumXiYi}
\end{align}
where $\mu \in \P_2(\RR^3)$ is any distribution with $M_\mu := \int v \mu(dv) = 0$
and $S_\mu^2 := \int |v-M_\mu|^2 \mu(dv) = 1$
(in general, $(S_\mu - S_\nu)^2 + |M_\mu-M_\nu|^2 \leq \W_2^2(\mu,\nu)$,
since for $X\sim\mu$ and $Y\sim\nu$ one has $\EE|X-Y|^2 = \EE[|(X-M_\mu) - (Y- M_\nu)|^2 + |M_\mu-M_\nu|^2
\geq S_\mu^2 + S_\nu^2 - 2S_\mu S_\nu + |M_\mu-M_\nu|^2$).
Since $\W_2^2(F,\st{F}) \leq \EE \frac{1}{N} \sum_i |X^i-Y^i|^2$, this
gives for any such $\mu$:
\begin{equation}
\label{eq:W2FNstFN}
\W_2^2(F,\st{F})
\leq \EE \W_2^2(\bar{\mathbf{X}}, \mu)
= \mathscr{W}_2^2(\hat{F},\delta_\mu).
\end{equation}

\begin{rmk}
\label{rmk:gaussian}
If $\gamma$ denotes the Gaussian density with mean $0$ and variance $1$,
that is, $\gamma(v) = (2\pi\sigma^2)^{-3/2} e^{-|v|^2/(2\sigma^2)}$ for $\sigma^2 = 1/3$,
then the measure $\st{\gamma^{\otimes N}}$ corresponds to $\mathcal{U}^N$. Thus, from \eqref{eq:W2FNstFN} applied to $F=\gamma^{\otimes N}$ and $\mu = \gamma$, we obtain $\W_2^2(\gamma^{\otimes N}, \mathcal{U}^N) \leq \mathscr{W}_2^2(\widehat{\gamma^{\otimes N}}, \delta_\gamma)$. Since $\mathscr{W}_2^2(\widehat{\gamma^{\otimes N}}, \hat{\mathcal{U}}^N) \leq \W_2^2(\gamma^{\otimes N},\mathcal{U}^N)$ and using \eqref{eq:epsnmu} on $\mathscr{W}_2^2(\widehat{\gamma^{\otimes N}}, \delta_\gamma) = \varepsilon_N(\gamma)$, we deduce the following chaos rate for $\mathcal{U}^N$ (already established in \cite[Lemma 25-(i)]{fournier-guillin2015}; see \cite[Theorem 4.4]{hauray-mischler2014} for related $\W_1$ estimates):
\begin{equation}
\label{eq:W2Uhatdeltag}
\mathscr{W}_2^2(\hat{\mathcal{U}}^N , \delta_\gamma)
\leq 2 \mathscr{W}_2^2(\hat{\mathcal{U}}^N , \widehat{\gamma^{\otimes N}}) + 2 \mathscr{W}_2^2(\widehat{\gamma^{\otimes N}}, \delta_\gamma)
\leq 4 \mathscr{W}_2^2(\widehat{\gamma^{\otimes N}},\delta_\gamma)
\leq C N^{-1/2}.
\end{equation}
\end{rmk}

\begin{proof}[Proof of Theorem \ref{thm:uniform_Maxwell}]
Let us first prove the result in the case where $G_0^N$ is concentrated
on the Boltzmann sphere $\mathcal{S}^N$.
Noting that in this case $\mathscr{W}_2^2(\hat{G}_0^N,\hat{\mathcal{U}}^N) \leq 4$,
Lemma \ref{lem:rousset} gives
\begin{equation}
\label{eq:W2symG0Ginfty}
\mathscr{W}_2^2(\hat{G}_t^N,\hat{\mathcal{U}}^N)
\leq C_{p,\delta} (1+t)^{-\delta},
\end{equation}
for all $\delta < p-2$, where $C_{p,\delta}$ depends only on $p$, $\delta$
and $R_p := \sup_N \EE|V_0^1|^p$.
With this, from \eqref{eq:EW2W2hatF} we have:
\[
%\begin{equation}
%\label{eq:W2_3terms}
\EE \W_2^2(\bar{\mathbf{V}}_t,f_t)
= \mathscr{W}_2^2(\hat{G}_t^N, \delta_{f_t})
\leq C \left[ \mathscr{W}_2^2(\hat{G}_t^N, \hat{\mathcal{U}}^N)
     + \mathscr{W}_2^2(\hat{\mathcal{U}}^N,\delta_\gamma)
     + \mathscr{W}_2^2(\delta_\gamma,\delta_{f_t}) \right].
%\end{equation}
\]
The first and second terms are controlled using \eqref{eq:W2symG0Ginfty} and \eqref{eq:W2Uhatdeltag}, respectively. The third term is equal to $\W_2^2(\gamma,f_t)$, which, by Theorem 5.8 of
\cite{gabetta-toscani-wennberg1995}, converges exponentially  fast to  $0$  under our assumption 
that $f_0$ has finite $p_0$-moment for some $p_0>4$ (condition \eqref{eq:p0}). Thus, we can just  bound that term by $C_{p,\delta}(1+t)^{-\delta}$.  All this yields
\[
\EE \W_2^2(\bar{\mathbf{V}}_t,f_t)
\leq C_{p,\delta} (1+t)^{-\delta} + CN^{-1/3}.
\]
For $\eta \in (0,1/3)$ to be chosen, set $\bar{t} = N^{\eta/\delta}-1$,
so the last inequality implies that $\EE \W_2^2(\bar{\mathbf{V}}_t,f_t)\leq C_{p,\delta} N^{-\eta}$
for all $t>\bar{t}$, whereas Theorem \ref{thm:main} gives
$\EE \W_2^2(\bar{\mathbf{V}}_t,f_t) \leq Ch_0 + C N^{-1/3 + 2\eta/\delta}$ for all $t\leq \bar{t}$,
where $h_0 = \W_2^2(G_0^N,f_0^{\otimes N})$.
Setting $\eta = [3(1+2/\delta)]^{-1} < (p-2)/3p$ gives
$\EE \W_2^2(\bar{\mathbf{V}}_t,f_t) \leq Ch_0 + C_{p,\delta} N^{-\eta}$ for all $t\geq 0$.
Since $\delta$ can be chosen arbitrarily close to $p-2$, the result follows
in this case.

Now we prove the general case by reducing it to the previous one.
Consider the process $\mathbf{W}_t := \st{\mathbf{V}_t}$,
with the notation of $\eqref{eq:Yi}$. 
Set $M_t = \frac{1}{N} \sum_j V_t^j = M_0$, $S_t^2 = \frac{1}{N} \sum_j |V_t^j - M_0|^2=S_0^2$.
It holds that on the event $\{S_0>0\}$, $\mathbf{W}$ solves \eqref{eq:PS}
with the same Poisson measure associated to $\mathbf{V}$,
but starting with initial condition $\mathbf{W}_0 = \st{\mathbf{V}_0}$.
Specifically: given $v,v_*,m\in\RR^3$ and $s>0$, the homogeneity
of $\ii$ and $\jj$ and the definitions of \eqref{eq:Gamma_a_c} imply that for $\tilde{v} = (v-m)/s$
and $\tilde{v}_* = (v_*-m)/s$ we have $c(\tilde{v},\tilde{v}_*,z,\phi)
= c(v,v_*,z,\phi) / s$, and then
\[
\tilde{v}'
= \tilde{v} + c(\tilde{v},\tilde{v}_*, z, \phi)
%= \frac{v-m}{s} + \frac{c(v,v_*, z, \phi)}{s}
= \frac{v+c(v,v_*, z, \phi)-m}{s}
= \frac{v'-m}{s},
\]
and the same for $\tilde{v}_*'$.
This means that the standarization procedure is preserved by the collisions.
Since the function $c$ is the one involved
in the definition of the particle system,
this shows that $\mathbf{W}$ solves \eqref{eq:PS}, as desired.

Define the event $D = \{S_0^2 \geq 1/4\}$.
We have
\begin{equation}
\label{eq:W2Vtft}
\EE \W_2^2(\bar{\mathbf{V}}_t,f_t)
\leq 2 \EE \W_2^2(\bar{\mathbf{V}}_t, \bar{\mathbf{W}}_t)
    + 2 \EE \ind_D \W_2^2(\bar{\mathbf{W}}_t, f_t)
    + \EE \ind_{D^c} \W_2^2(\bar{\mathbf{V}}_t, f_t).
\end{equation}
From \eqref{eq:sumXiYi} we have
$\W_2^2(\bar{\mathbf{V}}_t, \bar{\mathbf{W}}_t)
\leq \frac{1}{N} \sum_i |V_t^i - W_t^i|^2 = (S_t-1)^2 + |M_t|^2
= (S_0-1)^2 + |M_0|^2 \leq \W_2^2(\bar{\mathbf{V}}_0, f_0)$,
since $\int v f_0(dv) = 0$ and $\int |v|^2 f_0(dv) = 1$
Thus, $\W_2^2(\bar{\mathbf{V}}_t, \bar{\mathbf{W}}_t) \leq
\W_2^2(\bar{\mathbf{V}}_0, f_0) \leq 2\W_2^2(\bar{\mathbf{V}}_0, \bar{\mathbf{U}}_0)
     + 2 \W_2^2(\bar{\mathbf{U}}_0, f_0)$, and then for the first term of \eqref{eq:W2Vtft} we have
\begin{equation}
\label{eq:W2VtVtildet}
\EE \W_2^2(\bar{\mathbf{V}}_t, \bar{\mathbf{W}}_t)
%\leq \EE \frac{1}{N} \sum_i |W_t^i - V_t^i|^2
%\leq \EE [ (S_0-1)^2 + |M_0|^2 ]
%\leq \EE \W_2^2(\bar{\mathbf{V}}_0, f_0)
%\leq 2\EE\W_2^2(\bar{\mathbf{V}}_0, \bar{\mathbf{U}}_0)
%     + 2\EE\W_2^2(\bar{\mathbf{U}}_0, f_0)
\leq 2\EE \frac{1}{N} \sum_i |V_0^i - U_0^i|^2 + 2\EE\W_2^2(\bar{\mathbf{U}}_0, f_0)
\leq 2 \W_2^2(G_0^N,f_0^{\otimes N}) + C N^{-1/2},
\end{equation}
where we have used \eqref{eq:epsnmu}.

On $D$, we have $|W_0^i| = |V_0^i + M_0|/S_0 \leq 2(|V_0^i| + |M_0|)$,
and thus $\EE(\ind_D |W_0^i|^p) \leq C_p \EE|V_0^i|^p$.
Since $\mathbf{W}$ is a particle system taking values on $\mathcal{S}^N$,
we can apply the previous case, obtaining
$\EE \ind_D \W_2^2(\bar{\mathbf{W}}_t,f_t) \leq C\W_2^2(\st{G_0^N},f_0^{\otimes N})
                          + C_{p,\epsilon} N^{-(p-2)/3p + \epsilon}$.
Using \eqref{eq:W2FNstFN} we have $\W_2^2(\st{G_0^N},f_0^{\otimes N})
\leq 2 \W_2^2(\st{G_0^N},G_0^N) + 2 \W_2^2(G_0^N,f_0^{\otimes N})
\leq 2 \EE \W_2^2(\bar{\mathbf{V}}_0,f_0) + 2 \W_2^2(G_0^N,f_0^{\otimes N})$,
and from there, the same argument used to estimate the first term of \eqref{eq:W2Vtft}
yields for the second term:
\begin{equation}
\label{eq:W2tildeVtft}
\EE \ind_D \W_2^2(\bar{\mathbf{W}}_t,f_t) \leq C\W_2^2(G_0^N,f_0^{\otimes N})
                          + C_{p,\epsilon} N^{-(p-2)/3p + \epsilon}.
\end{equation}

For the third term, using the preservation of momentum and energy, we have
\begin{align*}
\EE \ind_{D^c} \W_2^2(\bar{\mathbf{V}}_t,f_t)
&\leq \EE \left( \ind_{\{S_0^2 < 1/4\}} 2\left[\frac{1}{N}\sum_i |V_t^i|^2 + \int |v|^2 f_t(dv)\right] \right) \\
&= \EE \left( \ind_{\{S_0^2 < 1/4\}} 2[S_0^2 + |M_0|^2 + 1]\right) \\
&\leq \frac{10}{4} \PP(S_0^2 < 1/4) + 2 \EE |M_0|^2,
\end{align*}
and since $\PP(S_0^2 < 1/4) \leq \PP(|S_0 - 1| > 1/2) \leq 4 \EE|S_0 - 1|^2$,
we have shown that $\EE \ind_{D^c} \W_2^2(\bar{\mathbf{V}}_t,f_t) \leq 10 \EE(|S_0 - 1|^2 + |M_0|^2)$,
which is again controlled by $10 \EE\W_2^2(\bar{\mathbf{V}}_0,f_0)
\leq C\W_2^2(G_0^N,f_0^{\otimes N}) + CN^{-1/2}$ as above.

Finally, putting the previous estimate, \eqref{eq:W2VtVtildet} and \eqref{eq:W2tildeVtft}
into \eqref{eq:W2Vtft}, yields the desired result.
\end{proof}

\section*{Appendix}

\begin{proof}[Proof of Lemma \ref{lemaSDE}]

Part \ref{lemaSDE_i} follows from the last step in the proof  of part \ref{lemaSDE_ii} below.

Since $L<\infty$, in order to  prove \ref{lemaSDE_ii} it is enough to construct a weak solution $\mathbf{V}$  to \eqref{eq:PS} and then build $\mathbf{U}^L $ driven by the same Poisson process in a jump-by-jump manner.

To obtain a weak solution of \eqref{eq:PS}, we will use a cutoff procedure: for a given cutoff level $K\in[1,\infty)$, define $\mathbf{V}^K = (V^{1,K},\ldots,V^{N,K})$ as the solution to
\begin{equation}
\label{eq:PSK}
d\mathbf{V}_t^K
= \int_0^\infty \int_0^{2\pi} \int_{[0,N)^2} \sum_{i\neq j} \ind_{\{\i(\xi) = i,\i(\zeta)=j\}}
 \mathbf{c}_{K,ij}(\mathbf{V}_{t^-}^K,z,\phi) \N(dt,dz,d\phi,d\xi,d\zeta)
\end{equation}
where $\mathbf{c}_{K,ij}$ is defined as in \eqref{eq:cij} but using $c_K$ in place of $c$. Again, since $K<\infty$, the system $\mathbf{V}^K$ can be constructed pathwise.
Thus, given  a sequence of finite cutoff levels $K\to \infty$, one can prove in a similar way as in  Proposition 1.2--(ii) of   \cite{fournier-mischler2016}  that  the laws of $\mathbf{V}^K$ are tight (the  second moment estimates are  indeed trivial here because of their 
 exact preservation  stated in  \ref{lemaSDE_i}). By martingale methods and  classic probability space enlargement arguments,  one then gets that the accumulation  points are weak solutions of \eqref{eq:PS}.   
 In order to prove uniqueness in law of weak solutions of  \eqref{eq:PS}  it is enough to show that any weak  solution  can be approximated, in a pathwise way as $K\to \infty$,  by (strong) solutions $\hat{\mathbf{V}}^K$  to \eqref{eq:PSK}  driven by  some  Poisson measures defined  in the same probability space. More specifically: given $\hat{\mathbf{V}}^{\infty}$ a weak solution to \eqref{eq:PS} driven by the Poisson point measure $\hat{\N}(dt,dz,d\phi,d\xi,d\zeta)$,   define $\hat{\mathbf{V}}^{K}$ uniquely by 
\begin{align*}
d\hat{\mathbf{V}}^{K}_t
&= \int_0^\infty \int_0^{2\pi} \int_{[0,N)^2} \sum_{i\neq j} \ind_{\{\i(\xi) = i,\i(\zeta)=j\}} \\
&\qquad\qquad\qquad {} \times \mathbf{c}_{K,ij}(\hat{\mathbf{V}}_{t^-}^{K},z,\phi_{t^-}^K) \hat{\N}(dt,dz,d\phi,d\xi,d\zeta)
\end{align*} 
with  $\hat{\mathbf{V}}^{K}_0=\hat{\mathbf{V}}_0$ and 
where $\phi_{t-}^K=\phi_{t-}^K(\phi,\xi,\zeta)= \varphi(\hat{V}_{t-}^{\i(\xi),\infty} - \hat{V}_{t-}^{\i(\zeta),\infty}, \hat{V}_{t-}^{\i(\xi),K} - \hat{V}_{t-}^{\i(\zeta),K}, \phi)$. Notice that $\hat{\mathbf{V}}^{K}$ has the same law as $\mathbf{V}^{K}$, by pathwise uniqueness  for \eqref{eq:PS} and the fact that  the point measure  $\hat{\N}^K$ defined on test functions  by  $f\mapsto\int f(t,z,\phi_{t^-}^K,d\xi,d\zeta)\hat{\N}(dt,dz,d\phi,d\xi,d\zeta)$ has the same law as $\N(dt,dz,d\phi,d\xi,d\zeta)$ (as can be checked using  It\^{o} calculus and Campbell's formula).  
Since the coupling between circles provided by Lemma \ref{lem:optimal_circles} is optimal, integrated versions of the bounds provided in Lemma 5.1 of \cite{fournier-mischler2016}  (for some other coupling of angles $\phi$)  also hold in the present context. 
Using It\^o  calculus and  Gronwall's lemma we are then  able to prove, in a similar way  that, for each $T>0$ and some function $R_T(K)$ going to $0$ as  $K\to\infty$ one has
$
\sup_{t\in [0,T]} \EE( |\hat{\mathbf{V}}^{K}_t- \hat{\mathbf{V}}^{\infty}_t|^2)\leq  R_T(K),
$
from where we get the desired convergence:
\[
\EE( \sup_{t\in [0,T]}|\hat{\mathbf{V}}^{K}_t- \hat{\mathbf{V}}^{\infty}_t|^2)\leq  R_T(K)
\] 
Part \ref{lemaSDE_ii} follows.

Now we prove part \ref{lemaSDE_iii}.
Define $\tilde{\N}^{i,L}(dt,dz,d\phi,dv)$
as the point measure on $[0,\infty)\times[0,\infty)\times[0,2\pi)\times \RR^3$
with atoms $(t,z,\varphi_{t^-}^i,\Pi_t^{i,L}(\mathbf{U}_{t^-}^L,\xi))$
for every atom $(t,z,\phi,\xi)$ of $\N^i$.
Since: 1) the dependence on $(\mathbf{V},\mathbf{U}^L)$ is predictable,
2) the $\xi$-law of $\Pi_t^{i,L}(\mathbf{x},\xi)$ is $f_t^L$ for
any $\mathbf{x}\in(\RR^3)^N$, and
3) the $\phi$-law of $\varphi(v,u,\phi)$ is the uniform distribution
on $[0,2\pi)$ for any $v,u\in\RR^3$,
one can compute the Laplace functional
of $\tilde{\N}^{i,L}$ and deduce that it is a Poisson point measure with
intensity $dtdzd\phi f_t^L(dv)/2\pi$. From \eqref{eq:Ui},
it is then clear that $U^{i,L}$ satisfies \eqref{eq:nonlinear_cutoff},
with $\M^L$ replaced by $\tilde{\N}^{i,L}$, which shows that $U^{i,L}$
is a cutoff nonlinear process.

 Part  \ref{lemaSDE_iv} is obvious.
\end{proof}

\begin{proof}[Proof of Lemma \ref{lem:Povzner}, case of $p$ even]
We will use the parametrization of $\SS^2$ given by \eqref{eq:Gamma_a_c},
that is, $v' = v + a(v,v_*,\theta,\phi)$ and $v_*' = v_*-a(v,v_*,\theta,\phi)$.
For notational simplicity, call $c := \frac{1-\cos\theta}{2}$, $s = \frac{\sin\theta}{2}$,
$\ii = \ii(v-v_*)$, $\jj = \jj(v-v_*)$ and $\Gamma = \Gamma(v-v_*,\phi)$
(recall that $\Gamma = (\cos\phi)\ii + (\sin\phi)\jj$).
Noting that $|\Gamma|^2 = |v-v_*|^2$ and that $(v-v_*)\cdot \Gamma = 0$, we have 
\begin{align*}
|v'|^2
%&= |v + a(v,v_*,\theta,\phi)|^2 \\
&= |v + c (v_*-v) + s \Gamma|^2 \\
%&= |v|^2 + c^2 |v_*-v|^2 + s^2 |v_*-v|^2 + 2c v\cdot (v_*-v) + 2sv \cdot \Gamma \\
&= [1+c^2+s^2-2c]|v|^2 + [c^2 +s^2]|v_*|^2 + 2[c - c^2 - s^2] v\cdot v_* + 2s v\cdot \Gamma \\
&= (1-c)|v|^2 + c|v_*|^2 + s (v+v_*)\cdot \Gamma,
\end{align*}
where in the last step we used the identity $c^2 +s^2 = c$ and the fact that $v\cdot \Gamma = v_* \cdot \Gamma$.
Calling $w=v+v_*$, we thus obtain
\begin{align*}
|v'|^2 &= x + s w\cdot \Gamma, \qquad \text{for $x = (1-c)|v|^2 + c|v_*|^2$, and} \\
|v_*'|^2 &= y - s w\cdot\Gamma, \qquad \text{for $y = (1-c)|v_*|^2 + c|v|^2$,}
\end{align*}
where the second identity is deduced similarly as the first one. Take $p=2k$ for
integer $k\geq 2$. Thus,
\begin{equation}
\label{eq:int_v'p_phi}
\int_0^{2\pi} |v'|^p \frac{d\phi}{2\pi}
= x^k + \sum_{i=1}^{\lfloor k/2 \rfloor} \binom{k}{2i} x^{k-2i} s^{2i} \int_0^{2\pi} (w\cdot \Gamma)^{2i} \frac{d\phi}{2\pi},
\end{equation}
where we have used the fact that $\int_0^{2\pi} (w\cdot \Gamma)^{2i-1} = 0$, since it is
computed as the sum of integrals of terms of the form $(\cos \phi)^a (\sin\phi)^b$ where $a$ or $b$ is odd.
The key of the proof is to show that, after integration in $\phi$, the term $(w\cdot \Gamma)^{2i}$
is of order $|v|^{2i}|v_*|^{2i}$ and not $|v|^{4i}|v_*|^{4i}$, as one would
obtain using loose bounds.
Specifically: using the same argument to neglect the odd terms of the sum, we have
\begin{align*}
\int_0^{2\pi} (w\cdot \Gamma)^{2i} \frac{d\phi}{2\pi}
&=\int_0^{2\pi} [(\cos\phi)w\cdot \ii + (\sin\phi)w\cdot\jj]^{2i} \frac{d\phi}{2\pi} \\
&= \sum_{j=0}^i \binom{2i}{2j} (w\cdot \ii)^{2i-2j} (w\cdot \jj)^{2j}
               \int_0^{2\pi} (\cos\phi)^{2i-2j} (\sin\phi)^{2j}\frac{d\phi}{2\pi}.
\end{align*}
Denoting by $n!!$  the product of the positive integers smaller than or equal to $n$ which have the
same parity as $n$,  one can check the identity $\int_0^{2\pi} (\cos\phi)^{n} (\sin\phi)^{m}\frac{d\phi}{2\pi} = \frac{(n-1)!! (m-1)!!}{(n+m)!!}$ for $n$ and $m$ even (integrate by parts in both possible ways,   use  $\cos^2\phi+ \sin^2\phi=1$ to get two recurrence relations and deduce the identity by double induction in $(n,m)$).
Since $(2n-1)!! = \frac{(2n)!}{2^n n!}$ and $(2n)!! = 2^n n!$,
it can be easily seen that $\binom{2i}{2j} \frac{(2i-2j-1)!! (2j-1)!!}{(2i)!!} = 2^{-2i} \binom{2i}{i}\binom{i}{j}$,
which yields
\[
\int_0^{2\pi} (w\cdot \Gamma)^{2i} \frac{d\phi}{2\pi}
= 2^{-2i} \binom{2i}{i} \sum_{j=0}^i \binom{i}{j} (w\cdot \ii)^{2i-2j} (w\cdot \jj)^{2j} 
= 2^{-2i} \binom{2i}{i} \left[ (w\cdot \ii)^2 + (w\cdot \jj)^2 \right]^{i}.
\]
Using that $(\frac{v-v_*}{|v-v_*|},\frac{\ii}{|v-v_*|},\frac{\jj}{|v-v_*|})$
is an orthonormal basis, we have
\begin{align*}
(w\cdot \ii)^2 + (w\cdot \jj)^2
&= |v-v_*|^2|w|^2 - (w\cdot(v-v_*))^2 \\
&= (|v|^2 + |v_*|^2 - 2v\cdot v_*) (|v|^2 + |v_*|^2 + 2v\cdot v_*)  - (|v|^2 - |v_*|^2)^2 \\
&= 4 |v|^2|v_*|^2 - 4(v\cdot v_*)^2,
\end{align*}
where the cancelation of $|v|^4$ and $|v_*|^4$ in the last step is the
crucial point of the proof. We deduce that
\[
\int_0^{2\pi} (w\cdot \Gamma)^{2i} \frac{d\phi}{2\pi}
= \binom{2i}{i} \left[ |v|^2|v_*|^2 - (v\cdot v_*)^2 \right]^{i}
\leq 2^i \binom{2i}{i} |v|^{2i}|v_*|^{2i}
\]
Denote $\tilde{A}_p>0$ some constant that depends only on $p$ and that can change from
line to line. With the last inequality and using the bound $x^{k-2i} \leq \tilde{A}_p (|v|^{2k-4i} + |v_*|^{2k-4i})$,
from \eqref{eq:int_v'p_phi} we obtain
\begin{align*}
\int_0^{2\pi} |v'|^p \frac{d\phi}{2\pi}
&\leq c^k|v|^{2k} + (1-c)^k|v_*|^{2k} + \sum_{i=1}^{k-1} \binom{k}{i} (1-c)^{k-i}|v|^{2k-2i} c^i|v_*|^{2i} \\
&\qquad \qquad \qquad {} + \tilde{A}_p \sum_{i=1}^{\lfloor k/2 \rfloor} s^{2i}
                           (|v|^{2k-2i}|v_*|^{2i} + |v_*|^{2k-2i}|v|^{2i}) \\
&\leq c^k|v|^{2k} + (1-c)^k|v_*|^{2k}
+ \theta^2 \tilde{A}_p (|v|^{2k-2}|v_*|^{2} + |v_*|^{2k-2}|v|^{2}),
\end{align*}
where in the last step we have used Young's inequality and the fact that $s^{2i}$
and $c^i$ are of order $\theta^2$ for $i\geq 1$.
With the same argument, the last inequality is valid replacing $v'$ for $v_*'$ and exchanging
the roles of $v$ and $v_*$. Since $c=\sin(\theta/2)^2$ and $1-c=\cos(\theta/2)^2$,
integrating $|v'|^p + |v_*'|^p - |v|^p - |v_*|^p$ against $\beta(\theta)d\theta \frac{d\phi}{2\pi}$
yields the result.
\end{proof}

\bigskip

{\bf Acknowledgments} R.C.'s  research  was partially supported by Mecesup UCH0607 Doctoral Fellowship,  BASAL-Conicyt Center for Mathematical Modeling  (CMM), ICM Millennium Nucleus NC120062 and Fondecyt Postdoctoral Proyect 3160250.  
 J.F.\ acknowledges partial support from  Fondecyt Project 1150570, ICM Millennium Nucleus NC120062 and  BASAL-Conicyt CMM. Both authors thank Nicolas Fournier for exchanges of mutual interest. 
We also thank anonymous referees for pointing out to us the reference \cite{carrapatoso2015}, for suggesting the use of the distance $\mathscr{W}_2$ to simplify the proof of Theorem \ref{thm:uniform_Maxwell} and for other suggestions that allowed us to improve the presentation of the paper.

%%%%%%%%%%%%%%%%%%%%%%%%%%%%%%%%%%%%%%%%%%%%%%%%%%%%%%%%
%%%%%%%%%%%%%%%%%%%%%%%%%%%%%%%%%%%%%%%%%%%%%%%%%%%%%%%%
\bibliographystyle{acm}
\bibliography{references_r1.bib}{}

\begin{thebibliography}{10}

\bibitem{alexandre2009}
{\sc Alexandre, R.}
\newblock A review of {B}oltzmann equation with singular kernels.
\newblock {\em Kinet. Relat. Models 2}, 4 (2009), 551--646.

\bibitem{ambrosio-gigli-savare2008}
{\sc Ambrosio, L., Gigli, N., and Savar{\'e}, G.}
\newblock {\em Gradient flows in metric spaces and in the space of probability
  measures}, second~ed.
\newblock Lectures in Mathematics ETH Z\"urich. Birkh\"auser Verlag, Basel,
  2008.

\bibitem{carrapatoso2015}
{\sc Carrapatoso, K.}
\newblock Quantitative and qualitative {K}ac's chaos on the {B}oltzmann's
  sphere.
\newblock {\em Ann. Inst. Henri Poincar\'e Probab. Stat. 51}, 3 (2015),
  993--1039.

\bibitem{cercignani1988}
{\sc Cercignani, C.}
\newblock {\em The {B}oltzmann equation and its applications}, vol.~67 of {\em
  Applied Mathematical Sciences}.
\newblock Springer-Verlag, New York, 1988.

\bibitem{cortez-fontbona2016}
{\sc Cortez, R., and Fontbona, J.}
\newblock Quantitative propagation of chaos for generalized {K}ac particle
  systems.
\newblock {\em Ann. Appl. Probab. 26}, 2 (2016), 892--916.

\bibitem{elmroth1983}
{\sc Elmroth, T.}
\newblock Global boundedness of moments of solutions of the {B}oltzmann
  equation for forces of infinite range.
\newblock {\em Arch. Rational Mech. Anal. 82}, 1 (1983), 1--12.

\bibitem{fontbona-guerin-meleard2009}
{\sc Fontbona, J., Gu{\'e}rin, H., and M{\'e}l{\'e}ard, S.}
\newblock Measurability of optimal transportation and convergence rate for
  {L}andau type interacting particle systems.
\newblock {\em Probab. Theory Related Fields 143}, 3-4 (2009), 329--351.

\bibitem{fontbona-guerin-meleard2010}
{\sc Fontbona, J., Gu\'erin, H., and M\'el\'eard, S.}
\newblock Measurability of optimal transportation and strong coupling of
  martingale measures.
\newblock {\em Electron. Commun. Probab. 15\/} (2010), 124--133.

\bibitem{fournier2015}
{\sc Fournier, N.}
\newblock Finiteness of entropy for the homogeneous boltzmann equation with
  measure initial condition.
\newblock {\em Ann. Appl. Probab. 25}, 2 (04 2015), 860--897.

\bibitem{fournier-guillin2015}
{\sc Fournier, N., and Guillin, A.}
\newblock From a {K}ac-like particle system to the {L}andau equation for hard
  potentials and {M}axwell molecules.
\newblock Preprint, arXiv:1510.01123, 2015.

\bibitem{fournier-guillin2013}
{\sc Fournier, N., and Guillin, A.}
\newblock On the rate of convergence in {W}asserstein distance of the empirical
  measure.
\newblock {\em Probability Theory and Related Fields 162}, 3-4 (2015),
  707--738.

\bibitem{fournier-meleard2002}
{\sc Fournier, N., and M{\'e}l{\'e}ard, S.}
\newblock A stochastic particle numerical method for 3{D} {B}oltzmann equations
  without cutoff.
\newblock {\em Math. Comp. 71}, 238 (2002), 583--604 (electronic).

\bibitem{fournier-mischler2016}
{\sc Fournier, N., and Mischler, S.}
\newblock Rate of convergence of the {N}anbu particle system for hard
  potentials and {M}axwell molecules.
\newblock {\em Ann. Probab. 44}, 1 (2016), 589--627.

\bibitem{gabetta-toscani-wennberg1995}
{\sc Gabetta, G., Toscani, G., and Wennberg, B.}
\newblock Metrics for probability distributions and the trend to equilibrium
  for solutions of the {B}oltzmann equation.
\newblock {\em J. Statist. Phys. 81}, 5-6 (1995), 901--934.

\bibitem{gamba-panferov-villani2004}
{\sc Gamba, I.~M., Panferov, V., and Villani, C.}
\newblock On the {B}oltzmann equation for diffusively excited granular media.
\newblock {\em Comm. Math. Phys. 246}, 3 (2004), 503--541.

\bibitem{graham-meleard1997}
{\sc Graham, C., and M{\'e}l{\'e}ard, S.}
\newblock Stochastic particle approximations for generalized {B}oltzmann models
  and convergence estimates.
\newblock {\em Ann. Probab. 25}, 1 (1997), 115--132.

\bibitem{grunbaum1971}
{\sc Gr{\"u}nbaum, F.~A.}
\newblock Propagation of chaos for the {B}oltzmann equation.
\newblock {\em Arch. Rational Mech. Anal. 42\/} (1971), 323--345.

\bibitem{hauray-mischler2014}
{\sc Hauray, M., and Mischler, S.}
\newblock On {K}ac's chaos and related problems.
\newblock {\em J. Funct. Anal. 266}, 10 (2014), 6055--6157.

\bibitem{kac1956}
{\sc Kac, M.}
\newblock Foundations of kinetic theory.
\newblock In {\em Proceedings of the {T}hird {B}erkeley {S}ymposium on
  {M}athematical {S}tatistics and {P}robability, 1954--1955, vol. {III}\/}
  (Berkeley and Los Angeles, 1956), University of California Press,
  pp.~171--197.

\bibitem{mckean1967}
{\sc McKean, Jr., H.~P.}
\newblock An exponential formula for solving {B}oltzmann's equation for a
  {M}axwellian gas.
\newblock {\em J. Combinatorial Theory 2\/} (1967), 358--382.

\bibitem{meleard1996}
{\sc M{\'e}l{\'e}ard, S.}
\newblock Asymptotic behaviour of some interacting particle systems;
  {M}c{K}ean-{V}lasov and {B}oltzmann models.
\newblock In {\em Probabilistic models for nonlinear partial differential
  equations ({M}ontecatini {T}erme, 1995)}, vol.~1627 of {\em Lecture Notes in
  Math.} Springer, Berlin, 1996, pp.~42--95.

\bibitem{mischler-mouhot2013}
{\sc Mischler, S., and Mouhot, C.}
\newblock Kac's program in kinetic theory.
\newblock {\em Invent. Math. 193}, 1 (2013), 1--147.

\bibitem{mischler-wennberg1999}
{\sc Mischler, S., and Wennberg, B.}
\newblock On the spatially homogeneous {B}oltzmann equation.
\newblock {\em Ann. Inst. H. Poincar\'e Anal. Non Lin\'eaire 16}, 4 (1999),
  467--501.

\bibitem{povzner1962}
{\sc Povzner, A.~J.}
\newblock On the {B}oltzmann equation in the kinetic theory of gases.
\newblock {\em Mat. Sb. (N.S.) 58 (100)\/} (1962), 65--86.

\bibitem{rousset2014}
{\sc Rousset, M.}
\newblock A {$N$}-uniform quantitative {T}anaka's theorem for the conservative
  {K}ac's {$N$}-particle system with {M}axwell molecules.
\newblock Preprint, arXiv:1407.1965, 2014.

\bibitem{sznitman1984}
{\sc Sznitman, A.-S.}
\newblock \'{E}quations de type de {B}oltzmann, spatialement homog\`enes.
\newblock {\em Z. Wahrsch. Verw. Gebiete 66}, 4 (1984), 559--592.

\bibitem{sznitman1989}
{\sc Sznitman, A.-S.}
\newblock Topics in propagation of chaos.
\newblock In {\em {\'E}cole d'{\'E}t{\'e} de {P}robabilit{\'e}s de
  {S}aint-{F}lour {XIX}---1989}, vol.~1464 of {\em Lecture Notes in Math.}
  Springer, Berlin, 1991, pp.~165--251.

\bibitem{tanaka1978}
{\sc Tanaka, H.}
\newblock On the uniqueness of {M}arkov process associated with the {B}oltzmann
  equation of {M}axwellian molecules.
\newblock In {\em Proceedings of the {I}nternational {S}ymposium on
  {S}tochastic {D}ifferential {E}quations ({R}es. {I}nst. {M}ath. {S}ci.,
  {K}yoto {U}niv., {K}yoto, 1976)\/} (1978), Wiley, New
  York-Chichester-Brisbane, pp.~409--425.

\bibitem{tanaka1979}
{\sc Tanaka, H.}
\newblock Probabilistic treatment of the {B}oltzmann equation of {M}axwellian
  molecules.
\newblock {\em Z. Wahrsch. Verw. Gebiete 46}, 1 (1978/79), 67--105.

\bibitem{toscani-villani1999}
{\sc Toscani, G., and Villani, C.}
\newblock Probability metrics and uniqueness of the solution to the {B}oltzmann
  equation for a {M}axwell gas.
\newblock {\em J. Statist. Phys. 94}, 3-4 (1999), 619--637.

\bibitem{villani2002}
{\sc Villani, C.}
\newblock A review of mathematical topics in collisional kinetic theory.
\newblock In {\em Handbook of mathematical fluid dynamics, {V}ol. {I}}.
  North-Holland, Amsterdam, 2002, pp.~71--305.

\bibitem{villani2009}
{\sc Villani, C.}
\newblock {\em Optimal transport, old and new}, vol.~338 of {\em Grundlehren
  der Mathematischen Wissenschaften [Fundamental Principles of Mathematical
  Sciences]}.
\newblock Springer-Verlag, Berlin, 2009.

\bibitem{wennberg1997}
{\sc Wennberg, B.}
\newblock Entropy dissipation and moment production for the {B}oltzmann
  equation.
\newblock {\em J. Statist. Phys. 86}, 5-6 (1997), 1053--1066.

\end{thebibliography}
%%%%%%%%%%%%%%%%%%%%%%%%%%%%%%%%%%%%%%%%%%%%%%%%%%%%%%%%
%%%%%%%%%%%%%%%%%%%%%%%%%%%%%%%%%%%%%%%%%%%%%%%%%%%%%%%%

\end{document}